\documentclass[12pt]{article}
\usepackage[english]{babel}
\usepackage{dsfont}
\usepackage{graphicx}
\usepackage{amsmath}
\usepackage{amssymb}
\usepackage{amsthm}
\usepackage{natbib}
\usepackage{geometry}
\usepackage{rotating}
\usepackage{setspace}
\usepackage{color}
\usepackage{bm}

\newcounter{bean}

\allowdisplaybreaks[2]

\theoremstyle{definition}

\theoremstyle{plain}
\newtheorem{theorem}{Theorem}
\newtheorem{lemma}{Lemma}

\newtheorem{ass}{Assumption}
\theoremstyle{remark}
\newtheorem{remark}{Remark}

\newcommand{\bY}{\mathbf{Y}}
\newcommand{\given}{\;\middle|\;}
\newcommand{\bX}{\mathbf{X}}
\newcommand{\mL}{\mathcal{L}}
\newcommand{\br}{\mathbf{r}}
\newcommand{\ind}{\mathds{1}} 
\renewcommand{\tilde }{\widetilde}

\renewcommand{\thesection}{\Roman{section}}

\title{Optimal Data Collection for Randomized\\ Control Trials}

\author{\textit{By} \textsc{Pedro Carneiro, Sokbae Lee, and Daniel Wilhelm}\thanks{Carneiro:
University College London, Institute for Fiscal Studies (IFS), and Centre for Microdata Methods
and Practice (CeMMAP); Lee: Columbia University, IFS and CeMMAP;
Wilhelm: University College London and CeMMAP. We thank Frank Diebold,
Kirill Evdokimov, Michal Kolesar, David McKenzie, Ulrich M\"uller, Imran Rasul, and participants at various seminars for helpful discussions.
An early version of this paper was presented at Columbia University and  Princeton University in September 2014,
and at New York University and University of Pennsylvania in December 2014.
This work was supported in part by the European Research Council (ERC-2014-CoG-646917-ROMIA)
and by  the UK Economic and Social Research Council (ESRC) through a grant (RES-589-28-0001)
to the ESRC CeMMAP.}}

\date{20 August 2016}

\begin{document}
\maketitle

\begin{abstract}
In a randomized control trial, the precision of an average treatment effect estimator and the power of the corresponding t-test can be
improved either by collecting data on additional individuals, or by collecting additional
covariates that predict the outcome variable. We propose the use of pre-experimental data
such as a census, or a household survey, to inform the choice of both the sample size and
the covariates to be collected. Our procedure seeks to minimize the resulting average
treatment effect estimator's mean squared error or the corresponding t-test's power, subject to the researcher's budget constraint.
We rely on a modification of an orthogonal greedy algorithm that is conceptually simple and easy to implement
in the presence of a large number of potential covariates, and does not require any
tuning parameters. In two empirical applications, we show that our procedure can lead to substantial gains of up to 58\%, measured either in terms of reductions in data collection costs or in terms of improvements in the precision of the treatment effect estimator. 

\medskip


\medskip


\end{abstract}


\newpage

\section{Introduction}

This paper is motivated by the observation that empirical research in economics increasingly
involves the collection of original data through laboratory or field experiments \citep[see,
e.g.][among others]{duflo2007, Banerjee:Duflo:11, BBR:11, List:11, List:Rasul:11,
Hamermesh:13}. This observation carries with it a call and an opportunity for research to
provide econometrically sound guidelines for data collection.

We consider the decision problem faced by a researcher designing the survey for a
randomized control trial (RCT). We assume that the goal of the researcher is to obtain precise estimates of the
average treatment effect and/or a powerful t-test of the hypothesis of no treatment effect, using the experimental data.\footnote{\citet{Tetenov:2015} provides a decision-theory-based rationale for using hypothesis tests in the RCT
and \citet{BCMS:2016} develop a theory of experimenters, focusing on the motivation of randomization among other things.}
Data collection is costly and the researcher is restricted by a budget,
which limits how much data can be collected. We focus on optimally trading off the number of individuals included in the RCT and the choice of covariates elicited as part of the data collection process.

There are, of course, other factors potentially influencing the choice of covariates to be collected in a survey for an RCT. For example, one may wish to learn about the mechanisms through which the RCT is operating, check whether treatment or control groups are balanced, or measure heterogeneity in the impacts of the intervention being tested.
In practice, researchers place implicit weights on each of the main objectives they consider when designing surveys, and consider informally the different trade-offs involved in their choices. We show that there is substantial value to making this decision process more rigorous and transparent through the use of data-driven tools that optimize a well-defined objective. Instead of attempting to formalize the whole research design process, we focus on one particular trade-off that we think is of first-order importance and particularly conducive to data-driven procedures.


We assume the researcher has access to
pre-experimental data from the population from which the experimental data will be drawn or at least from a population that shares similar second moments of the variables to be collected.
The data set includes all the potentially relevant variables that one would consider collecting
for the analysis of the experiment. 
%
The researcher faces a fixed budget for implementing the survey for the RCT. Given this budget, the researcher chooses the survey's sample size and set of covariates to optimize the resulting treatment effect estimator's precision and/or the corresponding t-test's power. This choice takes place before the implementation of the RCT and could, for example, be part of a pre-analysis plan in which, among other things, the researcher specifies outcomes of interest, covariates to be selected, and econometric techniques to be used.

In principle, the trade-offs involved in this choice involve basic economic reasoning. For each possible covariate, one should be comparing the marginal benefit and marginal cost of including it in the survey, which in turn, depend on all the other covariates included in the survey. As we discuss below, in simple settings it is possible to derive analytic and intuitive solutions to this problem. Although these are insightful, they only apply in unrealistic formulations of the problem. 

In general, for each covariate, there is a discrete choice of whether to include it or not, and for each possible sample size, one needs to consider all possible combinations of
covariates within the budget. This requires a solution to a computationally difficult combinatorial
optimization problem. This problem is especially challenging when the set of potential
variables to choose from is large, a case that is increasingly encountered in today's big
data environment. Fortunately, with the increased availability of high-dimensional data,
methods for the analysis of such data sets have received growing attention in several fields,
including economics \citep{BCC:14}. This literature makes available a rich set of new tools,
which can be adapted to our study of optimal survey design.

In this paper, we propose the use of a computationally attractive algorithm based on the
orthogonal greedy algorithm (OGA) -- also known as the orthogonal matching pursuit; see, for example, \citet{tropp2004greed} and  \cite{tropp2007signal} among many
others. To implement the OGA, it is necessary to specify the stopping rule, which in turn
generally requires a tuning parameter. One attractive feature of our algorithm is that, once
the budget constraint is given, there is no longer the need to choose a tuning parameter to
implement the proposed method, as the budget constraint plays the role of a stopping rule.
In other words, we develop an automated OGA that is tailored to our own decision
problem. Furthermore, it performs well even when there are a large number of potential
covariates in the pre-experimental data set.

There is a large and important body of literature on the design of experiments, starting with
\citet{Fisher:1935}. There also exists an extensive body of literature on sample size (power)
calculations; see, for example, \citet{McConnell2015} for a practical guide.
Both bodies of this literature are concerned with the precision of treatment effect estimates, but
neither addresses the problem that concerns us. For instance, \citet{McConnell2015} have
developed methods to choose the sample size when cost constraints are binding, but they
neither consider the issue of collecting covariates nor its trade-off with selecting the sample size.

Both our paper and the standard literature on power calculations rely on the availability of
information in pre-experimental data. The calculations we propose can be seen as a
substantive reformulation and extension of the more standard power calculations, which are
an important part of the design of any RCT.  When conducting
power calculations, one searches for the sample size that allows the researcher to detect a
particular effect size in the experiment. The role of covariates can be accounted for if one
has pre-defined the covariates that will be used in the experiment, and one knows (based on
some pre-experimental data) how they affect the outcome. Then, once the significance level
and power parameters are determined (specifying the type I and type II errors one is willing
to accept), all that matters is the impact of the sample size on the variance of the treatment
effect.

Suppose that, instead of asking what is the minimum sample size that allows us to detect a
given effect size, we asked instead how small an effect size we could detect with a particular
sample size (this amounts to a reversal of the usual power calculation). In this simple setting
with pre-defined covariates, the sample size would define a particular survey cost, and we would
essentially be asking about the minimum size of the variance of the treatment effect estimator that one
could obtain at this particular cost, which would lead to a question similar to the one asked
in this paper. Therefore, one simple way to describe our contribution is that we adapt and
extend the information in power calculations to account for the simultaneous selection of
covariates and sample size, explicitly considering the costs of data collection.

To illustrate the application of our method we examine two recent experiments for which
we have detailed knowledge of the process and costs of data collection. We ask two
questions. First, if there is a single hypothesis one wants to test in the experiment,
concerning the impact of the experimental treatment on one outcome of interest, what is the
optimal combination of covariate selection and sample size given by our method, and how
much of an improvement in the precision of the impact estimate can we obtain as a result?
Second, what are the minimum costs of obtaining the same precision of the treatment effect as
in the actual experiment, if one was to select covariates and sample size optimally (what we
call the ``equivalent budget'')?

We find from these two examples that by adopting optimal data collection rules, not only
can we achieve substantial increases in the precision of the estimates (statistical importance) for a given budget,
but we can also accomplish sizeable reductions in the equivalent budget
(economic importance). To illustrate the quantitative importance of the latter, we show that the optimal
selection of the set of covariates and the sample size leads to a reduction of about 45
percent (up to 58 percent) of the original budget in the first (second) example we consider, while
maintaining the same level of the statistical significance as in the original experiment.

To the best of our knowledge, no paper in the literature directly considers our data collection problem.
Some papers address related  but very different problems \citep[see][]{hahn2011we,
List-et-al:11,Bhattacharya:Dupas:2012, McKenzie2012,Dominitz:Manski:16}. They study
some issues of data measurement, budget allocation or efficient estimation; however, they
do not consider the simultaneous selection of the sample size and covariates for the RCTs as
in this paper. Because our problem is distinct from the problems studied in these papers, we
give a detailed comparison between our paper and the aforementioned papers in Section
\ref{sec:literature}.

More broadly, this paper is related to a recent emerging literature in economics that
emphasizes the importance of micro-level predictions and the usefulness of machine
learning for that purpose. For example, \citet{KLMZ:2015} argue that prediction problems
are abundant in economic policy analysis, and recent advances in machine learning can be
used to tackle those problems. Furthermore, our paper is related to the contemporaneous
debates on pre-analysis plans which demand, for example, the selection of sample sizes and covariates before
the implementation of an RCT; see, for example, \citet{Coffman:Niederle:15} and
\citet{Olken:15} for the advantages and limitations of the pre-analysis plans.

The remainder of the paper is organized as follows. In Section \ref{sec:problem}, we describe our
data collection problem in detail. In Section \ref{sec:algorithm}, we propose the use of a
simple algorithm based on the OGA. In Section  \ref{sec:cost}, we discuss the costs of data
collection in experiments. In Section \ref{sec:app}, we present two empirical applications,
in Section \ref{sec:literature},  we discuss the existing literature, and in Section \ref{sec:
conclusion}, we give concluding remarks.
Appendices provide details that are omitted from the main text.

\section{Data Collection Problem}\label{sec:problem}

Suppose we are planning an RCT in which we randomly assign individuals to either a treatment ($D=1$)
or a control group ($D=0$) with corresponding potential outcomes $Y_1$ and $Y_0$, respectively. After administering the treatment to the treatment group, we collect data on outcomes $Y$ for both groups so that $Y=DY_1 + (1-D)Y_0$. We also conduct a survey to collect data on a potentially very high-dimensional vector of covariates $Z$ (e.g. from a household survey covering demographics, social background, income etc.) that predicts potential outcomes. These covariates are a subset of the universe of predictors of potential outcomes, denoted by $X$. Random assignment of $D$ means that $D$ is independent of potential outcomes and of $X$.

Our goal is to estimate the average treatment effect $\beta_0 := E[Y_1-Y_0]$ as precisely as possible, where we measure precision by the finite sample mean-squared error (MSE) of the treatment effect estimator, and/or produce a powerful t-test of the hypothesis $H_0: \beta_0=0$. Instead of simply regressing $Y$ on $D$, we want to make use of the available covariates $Z$ to improve the precision of the resulting treatment effect estimator. Therefore, we consider estimating $\beta_0$ in the regression
\begin{equation}\label{eq: RA model}
	Y = \alpha_0 + \beta_0 D + \gamma_0'Z + U,
\end{equation}
%
%
where $(\alpha_0,\beta_0,\gamma_0')'$ is a vector of parameters to be estimated and $U$ is an
error term. The implementation of the RCT requires us to make two decisions that may have a significant impact on the estimation of and inference on the average treatment effect:
\begin{enumerate}
	\item Which covariates $Z$ should we select from the universe of potential predictors $X$?
	\item From how many individuals ($n$) should we collect data on $(Y, D, Z)$?
\end{enumerate}
Obviously, a large experimental sample size $n$ reduces the variance of the treatment effect estimator. Similarly, collecting more covariates, in particular strong predictors of potential outcomes, reduces the variance of the residual $U$ which, in turn, also improves the variance of the estimator. At the same time collecting data from more individuals and on more covariates is costly so that, given a finite budget, we want to find a combination of sample size $n$ and covariate selection $Z$ that leads to the most precise treatment effect estimator possible.

In this section, we propose a procedure to make this choice based on a pre-experimental data set on $Y$ and $X$, such as a pilot study or a census from the same population from which we plan to draw the RCT sample.\footnote{In fact, we do not need the populations to be identical, but only require second moments to be the same.} The combined data collection and estimation procedure can be
summarized as follows:
\begin{enumerate}
	\item Obtain pre-experimental data $\mathcal{S}_{\rm pre}$ on $(Y, X)$.
	\item Use data in $\mathcal{S}_{\rm pre}$ to select the covariates $Z$ and sample size $n$.
	\item Implement the RCT and collect the experimental data $\mathcal{S}_{\exp}$ on $(Y, D, Z)$.
	\item Estimate the average treatment effect using $\mathcal{S}_{\exp}$.
	\item Compute standard errors.
\end{enumerate}
We now describe the five steps listed above in more detail. The main component of our
procedure consists of a proposal for the optimal choice of $n$ and $Z$ in Step 2, which is
described more formally in Section~\ref{sec:algorithm}.

\setcounter{bean}{0}
\begin{center}
\begin{list}
{\textbf{Step \arabic{bean}}.}{\usecounter{bean}}
\item \textbf{Obtain pre-experimental data.} We assume the availability of data on
    outcomes  $Y\in
    \mathbb{R}$ and covariates $X\in \mathbb{R}^M$ for the population from which we
    plan to draw the experimental data. We denote the pre-experimental sample of size $N$ by
    $\mathcal{S}_{\rm pre} := \{Y_i,X_i\}_{i=1}^N$. Our
    framework allows the number of potential covariates, $M$, to be very large (possibly
    much larger than the sample size $N$).  Typical examples would be census data,
    household surveys, or data from other, similar experiments. Another possible
    candidate is a pilot experiment that was carried out before the larger-scale role out of the
    main experiment, provided that the sample size $N$ of the pilot study is large enough for
    our econometric analysis in Step 2.

\item \textbf{Optimal selection of covariates and sample size.} We want to use the
    pre-experimental data to choose the sample size, and which covariates should be in our
    survey. Let $S\in \{0,1\}^M$ be a vector of ones and zeros of the same dimension as $X$.
    We say that the $j$th covariate ($X^{(j)}$) is selected if $S_j=1$, and denote by $X_S$
    the subvector of $X$ containing elements that are selected by $S$. For example, consider
    $X = (X^{(1)}, X^{(2)}, X^{(3)})$ and $S = (1,0,1)$. Then $X_S = (X^{(1)},
    X^{(3)})$. For any vector of coefficients $\gamma\in\mathbb{R}^M$, let $\mathcal{I}(\gamma)\in\{0,1\}^M$ denote the nonzero elements of $\gamma$ and  $Y(\gamma):=Y - \gamma_{\mathcal{I}(\gamma)}'X_{\mathcal{I}(\gamma)}$. We can then  rewrite \eqref{eq: RA model} as
	\begin{equation}
		Y(\gamma) = \alpha_0 + \beta_0 D + U(\gamma),
	\end{equation}
	where $\gamma\in\mathbb{R}^M$ and $U(\gamma):=Y -\alpha_0-\beta_0 D -\gamma_{\mathcal{I}(\gamma)}'X_{\mathcal{I}(\gamma)}$. For a given $\gamma$ and sample size $n$, we denote by $\hat{\beta}(\gamma,n)$ the OLS estimator of $\beta_0$ in a regression of $Y(\gamma)$ on a constant and $D$, using a random sample $\{Y_i,D_i,X_i\}_{i=1}^n$. We also consider the two-sided\footnote{The same arguments in this paper straightforwardly carry over to a one-sided t-test.} t-test of
	$$H_0:\; \beta_0=0\qquad\text{vs.}\qquad H_1:\; \beta_0 \neq 0 $$
	using the t-statistic
	$$\hat{t}(\gamma,n) := \frac{\hat{\beta}(\gamma,n)}{\sigma(\gamma)/\sqrt{n \bar{D}_n(1-\bar{D}_n)}}, $$
	where $\sigma^2(\gamma) := Var(U(\gamma))$ is the residual variance and $\bar{D}_n := \sum_{i=1}^n D_i/n$ the number of individuals in the treatment group divided by the sample size $n$.

    Data collection is costly and therefore constrained by a budget of the
    form $c(S,n) \leq B$, where $c(S,n)$ are the costs of collecting  the variables given by
    selection $S$ from $n$ individuals, and $B$ is the researcher's budget.


We assume the researcher is interested in collecting data so as to ensure good statistical properties of the resulting treatment effect estimator and the corresponding t-test. We consider two criteria, the MSE of $\hat{\beta}(\gamma,n)$ and the power of the t-test that employs $\hat{t}(\gamma,n)$. We now briefly argue that minimizing the MSE of $\hat{\beta}(\gamma,n)$ and maximizing power of the t-test lead to equivalent optimization problems for selecting the optimal collection of covariates and sample size. Subsequently, we directly consider that optimization problem and the approximation of its solution, thereby transparently covering both objectives at the same time.

First, consider choosing the experimental sample size $n$ and the covariate selection $S$ so as
to minimize the finite sample MSE of $\hat{\beta}(\gamma,n)$, i.e., we want to choose $n$
and $\gamma$ to minimize
$$
MSE\left( \hat{\beta}(\gamma,n) \given D_1,\ldots,D_n\right) :=
E\left[ \left(\hat{\beta}(\gamma,n)- \beta_0\right)^2 \given D_1,\ldots,D_n\right].
$$
%
subject to the budget constraint.

\begin{ass}\label{ass: random sampling}
	(i) $\{(Y_i,X_i,D_i)\}_{i=1}^n$ is an i.i.d. sample from the distribution of $(Y,X,D)$ such that $D$ is completely randomized. (ii) $\text{Var} (U(\gamma)|D=1) = \text{Var} (U(\gamma) | D=0)$ for all $\gamma\in\mathbb{R}^M$.
\end{ass}

Part (i) of this assumption is standard. There are other  assignment mechanisms such as re-randomization, but we focus on the simplest case in the paper. Part (ii) is a homoskedasticity assumption that is common in standard power calculations and requires the residual variance to be the same across the treatment and control group. This assumption is satisfied, for example, when the treatment effect is constant across individuals in the experiment. If the researcher feels uncomfortable with this assumption, it is necessary to collect a pilot study that produces pre-experimental data from the joint distribution of $(D,X)$. The power of the homoskedasticity assumption is that, as we discuss in more detail below, data on $D$ is not required for the optimal choice of $n$ and $S$.

The following lemma characterizes the finite sample MSE of the estimator under the above assumption. 

\begin{lemma}\label{lem: equiv MSE}
Under Assumption~\ref{ass: random sampling}, for any $\gamma\in\mathbb{R}^M$,
\begin{equation}\label{eq: min MSE}
	MSE\left(\hat{\beta}(\gamma,n) \given D_1,\ldots,D_n\right) = \frac{\sigma^2(\gamma)}{n \bar{D}_n(1-\bar{D}_n)}.
\end{equation}
\end{lemma}

The proof of this Lemma can be found in the appendix. Note that for each ($\gamma, n$), the MSE is minimized by the equal splitting between the
treatment and control groups. Hence,  suppose that the treatment and control groups are of
exactly the same size (i.e., $\bar{D}_n = 0.5$).  By Lemma~\ref{lem: equiv MSE}, minimizing the
MSE of the treatment effect estimator subject to the budget constraint,
\begin{equation}\label{eq: population problem}
	\min_{n\in \mathbb{N}_+,\, \gamma\in \mathbb{R}^M} MSE\left(\hat{\beta}(\gamma,n) \given D_1,\ldots,D_n\right)\qquad
	\text{s.t.}\qquad  c(\mathcal{I}(\gamma),n) \leq B,
\end{equation}
is equivalent to minimizing the residual variance in a regression of $Y$ on $X$, divided by
the sample size,
\begin{equation}\label{eq: population problem in terms of variance}
	\min_{n\in \mathbb{N}_+,\, \gamma\in \mathbb{R}^M} \frac{\sigma^2(\gamma)}{n}\qquad
	\text{s.t.}\qquad  c(\mathcal{I}(\gamma),n) \leq B,
\end{equation}

Now, consider choosing the experimental sample size $n$ and the covariate selection $S$ so as to maximize power of the two-sided t-test based on $\hat{t}(\gamma,n)$. Denote by $c_{\alpha}$ and $\Phi(\cdot)$ the $\alpha$-quantile and cumulative distribution function of the standard normal distribution, respectively. The following lemma calculates the test's finite sample power under the assumption of joint normality of $Y$ and $X$.

\begin{lemma}\label{lem: equiv power}
Suppose Assumption~\ref{ass: random sampling} holds and that $(Y,X)$ are jointly normal. Then, for any $\alpha\in(0,1)$, $\beta\neq 0$, and $\gamma\in\mathbb{R}^M$,
\begin{multline*}\label{eq: power}
	P_{\beta}\left(\left|\hat{t}(\gamma,n)\right| > c_{1-\alpha/2} \given D_1,\ldots,D_n\right)\\
	= 1+ \Phi\left( \frac{ \beta}{\sigma(\gamma) / \sqrt{n\bar{D}_n(1-\bar{D}_n)}} -c_{1-\alpha/2}\right) - \Phi\left( \frac{ \beta}{\sigma(\gamma) / \sqrt{n\bar{D}_n(1-\bar{D}_n)}} +c_{1-\alpha/2}\right),
\end{multline*}
where $P_{\beta}$ denotes probabilities under the assumption that $\beta$ is the true coefficient in front of $D$. Furthermore, $P_{\beta}(|\hat{t}(\gamma,n)| > c_{1-\alpha/2} | D_1,\ldots,D_n)$ is decreasing in $\sigma(\gamma)/ \sqrt{n\bar{D}_n(1-\bar{D}_n)}$.
\end{lemma}

The lemma shows that, under the normality assumption and for any alternative $\beta\neq 0$ and size $\alpha$, the power of the two-sided t-test is a decreasing transformation of $\frac{\sigma^2(\gamma)}{n\bar{D}_n(1-\bar{D}_n)}$. Therefore, assigning as many individuals to the treatment as to the control group, besides minimizing the MSE above also maximizes power. Therefore, assuming again $\bar{D}_n=0.5$, maximizing power subject to the budget constraint,
\begin{equation*}
	\max_{n\in \mathbb{N}_+,\, \gamma\in \mathbb{R}^M} P_{\beta}\left(\left|\hat{t}(\gamma,n)\right| > c_{1-\alpha/2} \given D_1,\ldots,D_n\right)\qquad
	\text{s.t.}\qquad  c(\mathcal{I}(\gamma),n) \leq B,
\end{equation*}
is also equivalent to minimizing the residual variance in a regression of $Y$ on $X$, divided by
the sample size, as in \eqref{eq: population problem in terms of variance}. Notice that even when $(Y,X)$ are not jointly normal, the power expression in Lemma~\ref{lem: equiv power} may be approximately correct because the Berry-Esseen bound guarantees that the t-statistic's distribution is close to normal as long as $n$ is not too small.

Having motivated the optimization problem in \eqref{eq: population problem in terms of variance} in terms of minimization of the MSE of the treatment effect estimator as well as in terms of maximization of power of the corresponding t-test, we now discuss how to approximate the solution to \eqref{eq: population problem in terms of variance} in a given finite sample.

Importantly, notice that the optimization problem \eqref{eq: population problem in terms of variance} depends on the data only through the residual variance $\sigma^2(\gamma)$, which, under Assumption~\ref{ass: random sampling}, can be estimated before the randomization takes place, i.e. using the pre-experimental sample $\mathcal{S}_{\rm pre}$.
Therefore, employing the standard sample variance estimator of $\sigma^2(\gamma)$, the sample counterpart of our population optimization problem \eqref{eq: population problem in terms of variance} is
\begin{equation}\label{eq: sample problem}
	\min_{n\in \mathbb{N}_+,\, \gamma\in \mathbb{R}^M} \frac{1}{nN}\sum_{i=1}^N (Y_i-\gamma'X_i)^2\qquad \text{s.t.}\qquad  c(\mathcal{I}(\gamma),n) \leq B.
\end{equation}
The problem \eqref{eq: sample problem}, which is based on the pre-experimental sample, approximates the population problem \eqref{eq: population problem in terms of variance} for the experiment if the second moments in the pre-experimental sample are close to the second moments in the experiment (which holds, for example, if the population in the pre-experimental sample is the same as the population in the experiment).

In Section~\ref{sec:algorithm}, we describe a computationally attractive OGA that approximates the solution to \eqref{eq: sample problem}. The OGA has been studied extensively in the signal extraction literature and is
implemented in most statistical software packages.
 Appendices A and D show that this algorithm possesses desirable theoretical and practical properties.

The basic idea of the algorithm (in its simplest form) is straightforward. Fix a sample size
$n$. Start by finding the covariate that has the highest correlation with the outcome. Regress
the outcome on that variable, and keep the residual. Then, among the remaining covariates,
find the one that has the highest correlation with the residual. Regress the outcome onto
both selected covariates, and keep the residual. Again, among the remaining covariates, find
the one that has the highest correlation with the new residual, and proceed as before. We
iteratively select additional covariates up to the point when the budget constraint is no
longer satisfied. Finally, we repeat this search process for alternative sample sizes,
and search for the  combination of sample size and covariate selection that minimizes the
residual variance. Denote the OGA solution by $(\hat{n},\hat{\gamma})$ and let $\hat{\mathcal{I}}:= \mathcal{I}(\hat{\gamma})$ denote the selected covariates. See Section~\ref{sec:algorithm} for more details.

Note that, generally speaking, the OGA requires us to specify how to terminate the iterative procedure.
One attractive feature of our algorithm is that the budget constraint plays the role of
the stopping rule, without introducing any tuning parameters.

\item \textbf{Experiment and data collection.} Given the optimal selection of
    covariates $\hat{\mathcal{I}}$ and sample size
    $\hat{n}$, we randomly assign $\hat{n}$ individuals to
    either the treatment or the control group (with equal probability), and collect the
    covariates $Z:= X_{\hat{\mathcal{I}}}$ from each of them. This yields the experimental
    sample $\mathcal{S}_{\exp} := \{Y_i,D_i,Z_i\}_{i=1}^{\hat{n}}$ from
    $(Y,D,X_{\hat{\mathcal{I}}})$.

\item \textbf{Estimation of the average treatment effect.} We regress $Y_i$ on $(1, D_i, Z_i)$ using the experimental sample $\mathcal{S}_{\exp}$. The OLS estimator of the coefficient on $D_i$ is the average treatment effect estimator $\hat{\beta}$.

\item \textbf{Computation of standard errors.} Assuming the two samples
    $\mathcal{S}_{\rm pre}$ and $\mathcal{S}_{\exp}$ are independent, and that treatment
    is randomly assigned, the presence of the covariate selection Step 2 does not affect the asymptotic validity of the standard errors that one would use in the absence of Step 2. Therefore, asymptotically valid standard errors of $\hat{\beta}$ can be
    computed in the usual fashion \citep[see, e.g.,][]{Imbens:Rubin:2015}.
\end{list}
\end{center}

\subsection{Discussion}
\label{subsec: discussion}

In this subsection, we discuss some of conceptual and practical properties of our proposed data collection procedure.

\paragraph{Availability of Pre-Experimental Data.} As in standard power calculations, pre-experimental data provide essential
information for our procedure. The availability of such data is very common, ranging from census data sets and other household surveys to studies that were conducted in a similar context as the RCT we are planning to implement. In addition, if no such data set is available, one may consider running a pilot project that collects pre-experimental data.
We recognize that in some cases it might be difficult to
have the required information readily available. 
However, this is a problem that affects any attempt to a data-driven design of surveys, including standard power calculations. Even when pre-experimental data are imperfect, such calculations provide a valuable guide to survey design, as long as the
available pre-experimental data are not very different from the ideal data. In particular, our procedure only requires second moments of the pre-experimental variables to be similar to those in the population of interest.

\paragraph{The Optimization Problem in a Simplified Setup.} In general, the problem in \eqref{eq: sample problem} does not have a simple
solution and requires joint optimization problem over the sample size $n$ and the coefficient $\gamma$. To gain some intuition about the trade-offs in this problem, in Appendix C we consider a
simplified setup in which all covariates are orthogonal to each other, and the budget constraint has
a very simple form. In this case, the constraint can be substituted into the objective and the optimization
becomes univariate and unconstrained. We show that if all covariates have the same price, then one wants to
choose covariates up to the point where the percentage increase in survey costs equals the
percentage reduction in the residual variance from the last covariate. Furthermore, the elasticity of the
residual variance with respect to changes in sample size should equal the elasticity of the residual variance with
respect to an additional covariate. If the costs of data collection vary with covariates, then
this conclusion is slightly modified. If we organize variables by type according to their
contribution to the residual variance, then we want to choose variables of each type up to the point
where the percent marginal contribution of each variable to the residual variance equals its percent
marginal contribution to survey costs.

\paragraph{Imbalance and Re-randomization.}  In RCTs, covariates typically do not only serve as a means to improving the precision of treatment effect estimators, but also for checking whether the control and treatment groups are balanced. 
See, for example, \citet{Bruhn:McKenzie:09}
for practical issues concerning randomization and balance.
To rule out large biases due to imbalance, it is important to carry out balance checks for strong predictors of potential outcomes. Our procedure selects the strongest predictors as long as they are not too expensive (e.g. household survey questions such as gender, race, number of children etc.) and we can check balance for these covariates. However, in principle, it is possible that our procedure does not select a strong predictor that is very expensive (e.g. baseline test scores). Such a situation occurs in our second empirical application (Section~\ref{sec: school grants}). In this case, in Step 2, we recommend running the OGA a second time, forcing the inclusion of such expensive predictors. If the MSE of the resulting estimate is not much larger than that from the selection without the expensive predictor, then we may prefer the former selection to the latter so as to reduce the potential for bias due to imbalance at the expense of slightly larger variance of the treatment effect estimator.

An alternative approach to avoiding imbalance considers re-randomization until some criterion capturing the degree of balance is met (e.g., \citet{Bruhn:McKenzie:09}, \cite{morgan2012gf,Morgan2015re} and \cite{LDR:2016}). Our criterion for the covariate selection procedure in Step 2 can readily be adapted to this case; however, the details are not worked out here. It is an interesting future research topic to fully develop a data collection method for re-randomization based on the modified variance formulae in \cite{morgan2012gf}  and \cite{LDR:2016}, which account for the effect of re-randomization on the treatment effect estimator. 

\paragraph{Expensive, Strong Predictors.} When some covariates have similar predictive power, but respective prices that are substantially different, our covariate selection procedure may produce a suboptimal choice. For example, if the covariate with the highest price is also the most predictive, OGA selects it first even when there are other covariates that are much cheaper but only slightly less predictive. In Section \ref{sec: school grants}, we encounter an example of such a situation and propose a simple robustness check for whether removing an expensive, strong predictor may be beneficial.

\paragraph{Properties of the Treatment Effect Estimator.} Since the treatment indicator is assumed independent of $X$, standard asymptotic theory of the treatment effect estimator continues to hold for our estimator (despite the addition of a covariate selection step). For example, it is unbiased, consistent, asymptotically normal, and adding the covariates $X$ in the regression in \eqref{eq: RA model} cannot increase the asymptotic variance of the estimator. In fact, inclusion of a covariate strictly reduces the estimator's asymptotic variance as long as the corresponding true regression coefficient is not zero. All these results hold regardless of whether the true conditional expectation of $Y$ given $D$ and $X$ is in fact linear and additive separable as in \eqref{eq: RA model} or not. In particular, in some applications one may want to include interaction terms of $D$ and $X$ \citep[see, e.g.,][]{Imbens:Rubin:2015}. Finally, the treatment effect can be allowed to be heterogeneous (i.e. vary across individuals $i$) in which case our procedure estimates the average of those treatment effects.

%

\paragraph{An Alternative to Regression.} Step 4 consists of running the regression in \eqref{eq: RA model}. There are instances when it is desirable to modify this step. For example, if the selected sample size $\hat{n}$ is smaller than the number of selected covariates, then the regression in \eqref{eq: RA model} is not feasible. However, if the pre-experimental sample $\mathcal{S}_{\rm pre}$ is large enough, we can instead compute the OLS estimator $\hat{\gamma}$ from the regression of $Y$ on $X_{\hat{\mathcal{I}}}$ in $\mathcal{S}_{\rm pre}$. Then use $Y$ and $Z$ from the experimental sample $\mathcal{S}_{\rm exp}$ to construct the new outcome variable
    $\hat{Y}_i^\ast := Y_i - \hat{\gamma}'Z_i$ and compute the treatment effect
    estimator $\hat{\beta}$ from the regression of $\hat{Y}_i^\ast$ on $(1,D_i)$. This approach avoids fitting too many parameters when the experimental sample is small and has the additional desirable property that the resulting estimator is free from bias due to imbalance in the selected covariates.

\paragraph{Multivariate Outcomes.} It is straightforward to extend our data collection method to the case when there are multivariate outcomes. Appendix G provides details regarding how to deal with a vector of outcomes when we select the common set of regressors for all outcomes.

%

\section{A Simple Greedy Algorithm}\label{sec:algorithm}

In practice, the vector $X$ of potential covariates is typically high-dimensional, which
makes it challenging to solve the optimization problem \eqref{eq: sample problem}. In this
section, we propose a computationally feasible algorithm that is both conceptually simple
and performs well in our simulations. In particular, it requires only running many univariate, linear regressions
and can therefore easily be implemented in popular statistical packages such as STATA.

We split the joint optimization problem in \eqref{eq: sample problem} over $n$ and $\gamma$ into
two nested problems. The outer problem searches over the optimal sample size $n$, which
is restricted to be on a grid $n\in\mathcal{N}:=\{n_0,n_1,\ldots,n_K\}$, while the inner
problem determines the optimal selection of covariates for each sample size $n$:
\begin{equation}\label{eq: sample problem - nested}
	\min_{n\in\mathcal{N}} \frac{1}{n} \min_{\gamma\in \mathbb{R}^M} \frac{1}{N} \sum_{i=1}^N (Y_i-\gamma'X_i)^2\qquad \text{s.t.}\qquad c(\mathcal{I}(\gamma),n) \leq B.
\end{equation}
To convey our ideas in a simple form,  suppose for the moment that the budget constraint
has the following linear form,
\begin{align*}
c(\mathcal{I}(\gamma),n) = n \cdot |\mathcal{I}(\gamma)| \leq B,
\end{align*}
where $|\mathcal{I}(\gamma)|$ denotes the number of non-zero elements of $\gamma$.
Note that the budget constraint puts the restriction on the number of selected covariates,
that is, $|\mathcal{I}(\gamma)| \leq B/n$.

It is known to be  NP-hard (non-deterministic polynomial time hard) to find a solution to the
inner optimization problem in \eqref{eq: sample problem - nested} subject to the constraint
that $\gamma$ has $m$ non-zero components, also called an $m$-term approximation,
where $m$ is the integer part of $B/n$ in our problem. In other words, solving \eqref{eq:
sample problem - nested} directly is not feasible unless the dimension of covariates, $M$, is small \citep{natarajan1995tr,Davis:1997jk}.

There exists a class of computationally attractive procedures called greedy algorithms that
are able to approximate the infeasible solution. See \citet{Temlyakov:2011} for a detailed
discussion of greedy algorithms in the context of approximation theory.
\citet{tropp2004greed}, \citet{tropp2007signal}, \citet{barron2008kk}, \citet{Zhang:2009},
\citet{Huang-et-al:2011},  \citet{Ing:Lai:11}, and \citet{sancetta2016}, among many others,
demonstrate the usefulness of greedy algorithms for signal recovery in information theory,
and for the regression problem in statistical learning. We use a variant of OGA that can allow for selection of groups of variables (see, for example, \citet{Huang-et-al:2011}).


To formally define our proposed algorithm, we introduce some notation. For a vector $v$ of
$N$ observations $v_1,\ldots,v_N$, let $\|v\|_N := (1/N \sum_{i=1}^N v_i^2)^{1/2}$ denote
the empirical $L^2$-norm and let $\bY := (Y_1,\ldots,Y_N)'$.

Suppose that the covariates $X^{(j)}$, $j=1,\ldots,M$, are organized into $p$
pre-determined groups $X_{G_1}, \ldots, X_{G_p}$, where $G_k\subseteq \{1,\ldots,p\}$
indicates the covariates of group $k$. We denote the corresponding matrices of observations
by bold letters (i.e., $\bX_{G_k}$ is the $N \times |G_k|$ matrix of observations on
$X_{G_k}$, where $|G_k|$ denotes the number of elements of the index set $G_k$). By a
slight abuse of notation, we let $\bX_{k} := \bX_{\{k\}}$ be the column vector of
observations on $X_k$ when $k$ is a scalar. One important special case is that in which each group
consists of a single regressor. Furthermore, we allow for overlapping groups; in other words,
some elements can be included in multiple or even all groups. The group structure occurs naturally in experiments where data collection is carried out through surveys whose
questions can be grouped in those concerning income, those concerning education, and so
on.  This can also occur naturally when we consider multivariate outcomes. 
See Appendix G for details.

Suppose that the largest group size $J_{\max} := \max_{k=1,\ldots,p} |G_k|$ is small, so
that we can implement orthogonal transformations \emph{within each group} such that $(\bX_{G_j}'
\bX_{G_j})/N = \mathbf{I}_{|G_j|}$, where $\mathbf{I}_{d}$ is the $d$-dimensional
identity matrix. In what follows, assume that $(\bX_{G_j}' \bX_{G_j})/N =
\mathbf{I}_{|G_j|}$ without loss of generality. Let $|\cdot|_2$ denote the $\ell_2$ norm.
The following procedure describes our algorithm.

\setcounter{bean}{0}
\begin{center}
\begin{list}
{\textsc{Step} \arabic{bean}.}{\usecounter{bean}}
\item Set the initial sample size $n=n_0$.
\item Group OGA for a given sample size $n$:
\begin{enumerate}
\item[(a)] initialize the inner loop at $k=0$ and set the initial residual
    $\hat{\br}_{n,0}=\bY$, the initial covariate indices $\hat{\mathcal{I}}_{n,0} =
    \emptyset$ and the initial group indices $\hat{\mathcal{G}}_{n,0}=\emptyset$;
\item[(b)] separately regress $\hat{\br}_{n,k}$ on each group of regressors in
    $\{1,\ldots,p\}\backslash \hat{\mathcal{G}}_{n,k}$; call $\hat{j}_{n,k}$ the group of
    regressors with the largest $\ell_2$ regression coefficients,
$$ \hat{j}_{n,k} := \arg\max_{j\in \{1,\ldots,p\}\backslash
\hat{\mathcal{G}}_{n,k}}  \left| \bX_{G_j}' \hat{\br}_{n,k} \right|_2; $$ add
$\hat{j}_{n,k}$ to the set of selected groups, $\hat{\mathcal{G}}_{n,k+1} =
\hat{\mathcal{G}}_{n,k} \cup \{\hat{j}_{n,k}\}$;
\item[(c)] regress $\bY$ on the covariates $\bX_{\hat{\mathcal{I}}_{n,k+1}}$ where
    $\hat{\mathcal{I}}_{n,k+1} := \hat{\mathcal{I}}_{n,k} \cup G_{\hat{j}_{n,k}}$; call
    the regression coefficient
    $\hat{\gamma}_{n,k+1}:=(\bX_{\hat{\mathcal{I}}_{n,k+1}}'\bX_{\hat{\mathcal{I}}_{n,k+1}})^{-1}
    \bX_{\hat{\mathcal{I}}_{n,k+1}}'\bY$  and the residual $\hat{\br}_{n,k+1}:=\bY -
    \bX_{\hat{\mathcal{I}}_{n,k+1}}\hat{\gamma}_{n,k+1}$;
\item[(d)] increase $k$ by one and continue with (b) as long as
    $c(\hat{\mathcal{I}}_{n,k},n)\leq B$ is satisfied;
\item[(e)] let $k_n$ be the number of selected groups; call the resulting submatrix of
    selected regressors $\mathbf{Z}:=\bX_{\hat{\mathcal{I}}_{n,k_n}}$ and
    $\hat{\gamma}_{n} := \hat{\gamma}_{n,k_n}$, respectively.
\end{enumerate}
\item Set $n$ to the next sample size in $\mathcal{N}$, and go to Step 2 until (and
    including) $n=n_K$.
\item Set $\hat{n}$ as the sample size that minimizes the residual variance:
$$ \hat{n} := \arg\min_{n\in\mathcal{N}} \frac{1}{nN}\sum_{i=1}^N
\left(Y_i- \mathbf{Z}_{i}\hat{\gamma}_n\right)^2.$$
\end{list}
\end{center}

The algorithm above produces the selected sample size $\hat{n}$, the selection of covariates
$\hat{\mathcal{I}}:= \hat{\mathcal{I}}_{\hat{n},k_{\hat{n}}}$ with $k_{\hat{n}}$
selected groups and $\hat{m} := m(\hat{n}):= |\hat{\mathcal{I}}_{\hat{n},k_{\hat{n}}}|$
selected regressors. Here, $\hat{\gamma}:=\hat{\gamma}_{\hat{n}}$ is the corresponding
coefficient vector on the selected regressors $Z$.

\begin{remark}
Theorem~\ref{thm: risk bound} in Appendix~A gives a finite-sample bound on the criterion function resulting from our OGA method and, thus, also for the MSE of the resulting treatment effect estimator. The natural target
for this residual variance is an infeasible residual variance when $\gamma_0$ is known \emph{a priori}.
Theorem~\ref{thm: risk bound} establishes conditions under which the difference between
the residual variance resulting from our method and the infeasible residual variance decreases at a rate of $1/k$ as
$k$ increases, where $k$ is the number of the steps in the OGA. It is known in a simpler
setting than ours that this rate $1/k$ cannot generally be improved \citep[see,
e.g.,][]{barron2008kk}. In this sense, we show that our proposed method has a desirable
property. See Appendix~A for further details.
\end{remark}

\begin{remark}
There are many important reasons for collecting covariates, such as checking whether
randomization was carried out properly and identifying heterogeneous treatment effects,
among others. If a few covariates are essential for the analysis, we can guarantee their
selection by including them in every group $G_k$, $k=1,\ldots,p$.
\end{remark}

\begin{remark}
	In a simple model such as the one in Appendix C, the optimal combination of covariates equalizes the percent marginal contribution of an additional variable to the residual variance with the percent marginal contribution of the additional variable to the costs per interview. 
	Step 2 of the OGA selects the next covariate as the one that has the highest predictive power independent of its cost. Outside a class of very simple models as in Appendix C, it is difficult to determine an OGA approximation to the optimum that jointly takes into account both predictive power as it requires comparison of all possible covariate combinations. In our empirical application of Section~V.B, we study a case with heterogeneous costs and propose a sensitivity analysis that assesses whether the OGA solution significantly changes with perturbations of the set of potential covariates.
\end{remark}


\section{The Costs of Data Collection}\label{sec:cost}

In this section, we discuss the specification of the cost function $c(S,n)$ that defines the
budget constraint of the researcher. In principle, it is possible to construct a matrix
containing the value of the costs of data collection for every possible combination of $S$ and
$n$ without assuming any particular form of relationship between the individual entries.
However, determination of the costs for every possible combination of $S$ and $n$ is a
cumbersome and, in practice, probably infeasible exercise. Therefore, we consider the
specification of cost functions that capture the costs of all stages of the data collection
process in a more parsimonious fashion.

We propose to decompose the overall costs of data collection into three components:
administration costs $c_{\rm admin}(S)$, training costs $c_{\rm train}(S,n)$, and interview
costs $c_{\rm interv}(S,n)$, so that
\begin{equation}\label{eq: cost decomposition}
	c( S ,n)= c_{\rm admin}(S) + c_{\rm train}(S,n) + c_{\rm interv}(S,n).
\end{equation}
In the remainder of this section, we discuss possible specifications of the three types of costs
by considering fixed and variable cost components corresponding to the different stages of
the data collection process. The exact functional form assumptions are based on the
researcher's knowledge about the operational details of the survey process. Even though this
section's general discussion is driven by our experience in the empirical applications of
Section~\ref{sec:app}, the operational details are likely to be similar for many surveys, so
we expect the following discussion to provide a useful starting point for other data
collection projects.

We start by specifying survey time costs. Let $\tau_{j}$, $j=1,\ldots,M$, be the costs of
collecting variable $j$ for one individual, measured in units of survey time. Similarly, let
$\tau_0$ denote the costs of collecting the outcome variable, measured in units of survey
time. Then, the total time costs of surveying one individual to elicit the variables indicated by
$S$ are
$$T(S) := \tau_0 + \sum_{j=1}^{M}\tau_{j}S_{j}.$$

\subsection{Administration and Training Costs}\label{sub-sec:fixed-cost}

A data collection process typically incurs costs due to administrative work and training prior
to the start of the actual survey. Examples of such tasks are developing the questionnaire
and the program for data entry, piloting the questionnaire, developing the manual for
administration of the survey, and organizing the training required for the enumerators.

Fixed costs, which depend neither on the size of the survey nor on the sample size of survey
participants, can simply be subtracted from the budget. We assume that $B$ is already net
of such fixed costs.

Most administrative and training costs tend to vary with the size of the questionnaire and the
number of survey participants. Administrative tasks such as development of the
questionnaire, data entry, and training protocols are independent of the number of survey
participants, but depend on the size of the questionnaire (measured by the number of
positive entries in $S$) as smaller questionnaires are less expensive to prepare than larger
ones. We model those costs by
\begin{equation}\label{eq: admin}
	c_{\rm admin}(S) := \phi  T(S)^{\alpha},
\end{equation}
where $\phi$ and $\alpha$ are scalars to be chosen by the researcher. We assume
$0<\alpha<1$, which means that marginal costs are positive but decline with survey size.

Training of the enumerators depends on the survey size, because a longer survey requires
more training, and on the number of survey participants, because surveying more individuals
usually requires more enumerators (which, in turn, may raise the costs of training),
especially when there are limits on the duration of the fieldwork. We therefore specify
training costs as
 \begin{equation}\label{eq: train}
 	c_{\rm train}(S,n) := \kappa(n) \, T(S),
 \end{equation}
where $\kappa(n)$ is some function of the number of survey participants.\footnote{It is of
course possible that $\kappa$ depends not only on $n$ but also on $T(S)$. We model it this
way for simplicity, and because it is a sensible choice in the applications we discuss below.}
Training costs are typically lumpy because, for example, there exists only a limited set of
room sizes one can rent for the training, so we model $\kappa( n) $ as a step function:
\[
\kappa\left( n\right) =\left\{
\begin{array}{ccc}
\overline{\kappa}_{1} & \text{if} & 0<n\leq \overline{n}_{1} \\
\overline{\kappa}_{2} & \text{if} & \overline{n}_{1}<n\leq \overline{n}_{2}\\
&\vdots &
\end{array}%
\right. . \]
Here, $\overline{\kappa}_{1}, \overline{\kappa}_{2}, \ldots$ is a sequence of scalars
describing the costs of sample sizes in the ranges defined by the cut-off sequence
$\overline{n}_{1}, \overline{n}_{2}, \ldots$.

\subsection{Interview Costs}

Enumerators are often paid by the number of interviews conducted, and the payment
increases with the size of the questionnaire. Let $\eta$ denote fixed costs per interview that
are independent of the size of the questionnaire and of the number of participants. These are
often due to travel costs and can account for a substantive fraction of the total interview
costs. Suppose the variable component of the interview costs is linear so that total interview
costs can be written as
\begin{equation}\label{eq: interview}
	c_{\rm interv}(S,n) := n\eta + np\,T(S),
\end{equation}
where $T(S)$ should now be interpreted as the average time spent per interview,  and $p$ is
the average price of one unit of survey time. We employ the specification \eqref{eq: cost
decomposition} with \eqref{eq: admin}--\eqref{eq: interview} when studying the impact of
free day-care on child development in Section~\ref{sec: daycare}.

\begin{remark}
Because we always collect the outcome variable, we incur the fixed costs $n\eta$ and the
variable costs $np\tau_0$ even when no covariates are collected.
\end{remark}

\begin{remark}
Non-financial costs are difficult to model, but could in principle be added. They are
primarily related to the impact of sample and survey size on data quality. For example, if we
design a survey that takes more than four hours to complete, the quality of the resulting data
is likely to be affected by interviewer and interviewee fatigue. Similarly, conducting the
training of enumerators becomes more difficult as the survey size grows. Hiring high-quality
enumerators may be particularly important in that case, which could result in even higher
costs (although this latter observation could be explicitly considered in our framework).
\end{remark}

\subsection{Clusters}

In many experiments, randomization is carried out at a cluster level (e.g., school level),
rather than at an individual level (e.g., student level). In this case, training costs may depend
not only on the ultimate sample size $n = c\, n_c$, where $c$ and $n_c$ denote the number
of clusters  and the number of participants per cluster, respectively, but on a particular
combination ($c,n_c$), because the number of required enumerators may be different for
different ($c,n_c$) combinations. Therefore, training costs (which now also depend on $c$
and $n_c$) may be modeled as
\begin{equation}\label{eq: train cluster}
	c_{\rm train}(S,n_c,c) := \kappa(c, n_c)\, T(S).
\end{equation}
The interaction of cluster and sample size in determining the number of required
enumerators and, thus, the quantity $\kappa(c,n_c)$, complicates the modeling of this
quantity relative to the case without clustering. Let $\mu(c,n_c)$ denote the number of
required survey enumerators for $c$ clusters of size $n_c$. As in the case without clustering, we assume that the
training costs is lumpy in the number of enumerators used:
\[
\kappa( c,n_c) :=\left\{
\begin{array}{ccc}
\overline{\kappa}_{1} & \text{if} & 0< \mu(c,n_c)\leq \overline{\mu}_{1} \\
\overline{\kappa}_{2} & \text{if} & \overline{\mu}_{1}< \mu(c,n_c)\leq \overline{\mu}_{2}\\
&\vdots &
\end{array}%
\right.. \]
The number of enumerators required, $\mu(c,n_c)$, may also be lumpy in the number of
interviewees per cluster, $n_c$, because there are bounds to how many
interviews each enumerator can carry out. Also, the number of enumerators needed for the
survey typically increases in the number of clusters in the experiment. Therefore, we model
$\mu(c,n_c)$ as
$$\mu(c,n_c) := \lfloor\mu_c(c)\cdot \mu_n(n_c)\rfloor,$$
where $\lfloor\cdot \rfloor$ denotes the integer part, $\mu_c(c) := \lambda c $ for some
constant $\lambda$ (i.e., $\mu_c(c)$ is assumed to be linear in $c$), and
$$\mu_n(n_c) :=  \left\{
\begin{array}{ccc}
\overline{\mu}_{n,1} & \text{if} & 0< n_c\leq \overline{n}_{1} \\
\overline{\mu}_{n,2} & \text{if} & \overline{n}_{1}< n_c\leq \overline{n}_{2}\\
&\vdots &
\end{array}%
\right..$$

In addition, while the variable interview costs component continues to depend on the overall
sample size $n$ as in \eqref{eq: interview}, the fixed part of the interview costs is determined
by the number of clusters $c$ rather than by $n$. Therefore, the total costs per interview
become
\begin{equation}\label{eq: interview cluster}
	c_{\rm interv}(S,n_c,c) := \psi(c) \eta + c n_c p\, T(S),
\end{equation}
where $\psi(c)$ is some function of the number of clusters $c$.

\subsection{Covariates with Heterogeneous Prices}

In randomized experiments, the data collection process often differs across blocks of
covariates. For example, the researcher may want to collect outcomes of psychological tests
for the members of the household that is visited. These tests may need to be administered by
trained psychologists, whereas administering a questionnaire about background variables
such as household income, number of children, or parental education, may not require any
particular set of skills or qualifications other than the training provided as part of the data
collection project.

Partition the covariates into two blocks, a high-cost block (e.g., outcomes of psychological
tests) and a low-cost block (e.g., standard questionnaire). Order the covariates such that the
first $M_{\rm low}$ covariates belong to the low-cost block, and the remaining $M_{\rm
high}:=M-M_{\rm low}$ together with the outcome variable belong to the high-cost block.
Let
$$T_{\rm low}(S):= \sum_{j=1}^{M_{\rm low}}\tau_{j}S_{j}\qquad \text{and}\qquad
T_{\rm high}(S):= \tau_0 + \sum_{j=M_{\rm low}+1}^{M}\tau_{j}S_{j}$$
be the total time costs per individual of surveying all low-cost and high-cost covariates,
respectively. Then, the total time costs for all variables can be written as $T(S) = T_{\rm
low}(S)+T_{\rm high}(S)$.

Because we require two types of enumerators, one for the high-cost covariates and one for
the low-cost covariates, the financial costs of each interview (fixed and variable) may be
different for the two blocks of covariates. Denote these by $\psi_{\rm low}(c,n_c) \eta_{\rm
low} + c n_c p_{\rm low} T_{\rm low}(S)$ and $\psi_{\rm high}(c,n_c) \eta_{\rm high} +
c n_c p_{\rm high} T_{\rm high}(S)$, respectively.

The fixed costs for the high-cost block are incurred regardless of whether high-cost
covariates are selected or not, because we always collect the outcome variable, which here
is assumed to belong to this block. The fixed costs for the low-cost block, however, are
incurred only when at least one low-cost covariate is selected (i.e., when
$\sum_{j=1}^{M_{\rm low}} S_j > 0$). Therefore, the total interview costs for all
covariates can be written as
\begin{multline}\label{eq: interview blocked}
	c_{\rm interv}(S,n) := \ind\Big\{ \sum_{j=1}^{M_{\rm low}} S_j > 0\Big\}
   (\psi_{\rm low}(c,n_c) \eta_{\rm low} + c n_c p_{\rm low} T_{\rm low}(S))\\
		 + \psi_{\rm high}(c,n_c) \eta_{\rm high} + c n_c p_{\rm high} T_{\rm high}(S).
\end{multline}
The administration and training costs can also be assumed to differ for the two types of
enumerators. In that case,
\begin{align}
	c_{\rm admin}(S) &:= \phi_{\rm low}T_{\rm low}(S)^{\alpha_{\rm low}} + \phi_{\rm high}T_{\rm high}(S)^{\alpha_{\rm high}},\label{eq: admin blocked}\\
	c_{\rm train}(S,n) &:= \kappa_{\rm low}(c, n_c)\, T_{\rm low}(S) + \kappa_{\rm high}(c, n_c)\, T_{\rm high}(S)\label{eq: train blocked}.
\end{align}
We employ specification \eqref{eq: cost decomposition} with \eqref{eq: interview
cluster}--\eqref{eq: train blocked} when, in Section~\ref{sec: school grants}, we study the
impact on student learning of cash grants which are provided to schools.

\section{Empirical Applications}\label{sec:app}

\subsection{Access to Free Day-Care in Rio}
\label{sec: daycare}

In this section, we re-examine the experimental design of \citet{Attanasio:2014sf}, who
evaluate the impact of access to free day-care on child development and household
resources in Rio de Janeiro. In their dataset, access to care in public day-care centers, most
of which are located in slums, is allocated through a lottery, administered to children in the
waiting lists for each day-care center.

Just before the 2008 school year, children applying for a slot at a public day-care center
were put on a waiting list. At this time, children were between the ages of 0 and 3. For each
center, when the demand for day-care slots in a given age range exceeded the supply, the
slots were allocated using a lottery (for that particular age range). The use of such an
allocation mechanism means that we can analyze this intervention as if it was an RCT, where
the offer of free day-care slots is randomly allocated across potentially eligible recipients.
\citet{Attanasio:2014sf} compare the outcomes of children and their families who were
awarded a day-care slot through the lottery, with the outcomes of those not awarded a slot.

The data for the study were collected mainly during the second half of 2012, four and a half
years after the randomization took place. Most children were between the ages of 5 and 8. A
survey was conducted, which had two components: a household questionnaire, administered
to the mother or guardian of the child; and a battery of health and child development
assessments, administered to children. Each household was visited by a team of two field
workers, one for each component of the survey.

The child assessments took a little less than one hour to administer, and included five tests
per child, plus the measurement of height and weight. The household survey took between
one and a half and two hours, and included about 190 items, in addition to a long module
asking about day-care history, and the administration of a vocabulary test to the main carer
of each child.

As we explain below, we use the original sample, with the full set of items collected in the
survey, to calibrate the cost function for this example. However, when solving the survey design problem described in this paper we consider only a subset of
items of these data, with the original budget being scaled down properly. This is done for
simplicity, so that we can essentially ignore the fact that some variables are missing for part
of the sample, either because some items are not applicable to everyone in the sample, or
because of item non-response. We organize the child assessments into three
indices: cognitive tests, executive function tests, and anthropometrics (height and weight).
These three indices are the main outcome variables in the analysis. However, we use
only the cognitive tests and anthropometrics indices in our analysis, as we have fewer
observations for executive function tests.

We consider only 40 covariates out of the
total set of items on the questionnaire. The variables not included can be arranged into four groups: (i)
variables that can be seen as final outcomes, such as questions about the development and
the behavior of the children in the household; (ii) variables that can be seen as intermediate
outcomes, such as labor supply, income, expenditure, and investments in children; (iii)
variables for which there is an unusually large number of missing values; and (iv) variables
that are either part of the day-care history module, or the vocabulary test for the child's
carer (because these could have been affected by the lottery assigning children to day-care
vacancies).  We then drop four of the 40 covariates chosen, because their variance is zero in
the sample. The remaining $M=36$ covariates are related to the respondent's age, literacy,
educational attainment, household size, safety, burglary at home, day care, neighborhood,
characteristics of the respondent's home and its surroundings (the number of rooms, garbage
collection service, water filter, stove, refrigerator, freezer, washer, TV, computer, Internet,
phone, car, type of roof, public light in the street, pavement, etc.). We drop individuals for
whom at least one value in each of these covariates is missing, which leads us to use a
subsample with 1,330 individuals from the original experimental sample, which included 1,466 individuals.

\paragraph{Calibration of the cost function.}
We specify the cost function \eqref{eq: cost decomposition} with components \eqref{eq:
admin}--\eqref{eq: interview} to model the data collection procedure as implemented in
\citet{Attanasio:2014sf}. We calibrate the parameters using the actual budgets for training,
administrative, and interview costs in the authors' implementation. The contracted total budget
of the data collection process was R\$665,000.\footnote{There were some adjustments to the
budget during the period of fieldwork.}

For the calibration of the cost function, we use the originally planned budget of
R\$665,000, and the original sample size of 1,466. As mentioned above, there were $190$
variables collected in the household survey, together with a day-care module and a
vocabulary test. In total, this translates into a total of roughly 240 variables.\footnote{The
budget is for the 240 variables (or so) actually collected. In spite of that, we only use 36 of
these as covariates in this paper, as the remaining variables in the survey were not so much
covariates as they were measuring other intermediate and final outcomes of the experiment,
as we have explained before. The actual budget used in solving the survey design problem is
scaled down to match the use of only 36 covariates.} 
Appendix~B provides a detailed description of all components of the calibrated cost function.

\paragraph{Implementation.}

In implementing the OGA, we take each single variable as a possible group (i.e., each group
consists of a singleton set). We studentized all covariates to have variance one. To compare
the OGA with alternative approaches, we also consider LASSO and POST-LASSO for the inner optimization problem in Step 2 of our procedure. The LASSO solves
	\begin{align}\label{eq:lasso}
	\min_{\gamma} \frac{1}{N}\sum_{i=1}^N \left(Y_i-\gamma'X_i\right)^2 + \lambda \sum_j |\gamma_j |
	\end{align} with a tuning parameter $\lambda > 0$. The POST-LASSO procedure runs an OLS regression of $Y_i$ on the selected covariates (non-zero entries of $\gamma$) in \eqref{eq:lasso}. \cite{belloni2013}, for example, provide a detailed description of the two algorithms. It is known that
LASSO yields  biased regression coefficient estimates and that POST-LASSO can mitigate
this bias problem. Together with the outer optimization over the sample size using the LASSO or POST-LASSO solutions in the inner loop may lead to different selections of covariate-sample size combinations. This is because POST-LASSO re-estimates the regression equation which may lead to more precise estimates of $\gamma$ and thus result in a different estimate for the MSE of the treatment effect estimator.

In both LASSO implementations, the penalization parameter $\lambda$ is chosen so as to
satisfy the budget constraint as close to equality as possible. We start
with a large value for $\lambda$, which leads to a large penalty for non-zero entries in $\gamma$, so that few or no covariates are selected and the budget constraint holds. Similarly, we consider a very small value for $\lambda$ which leads to the selection of many covariates and violation of the budget. Then, we use a bisection algorithm to find the $\lambda$-value in this interval
for which the budget is satisfied within some pre-specified tolerance.

\begin{table}[ht]
\caption{Day-care (outcome: cognitive test)}\label{tab: daycare test1 summary}
 \begin{center}
\begin{tabular}{lcccccc}
\hline\hline
Method     & $\hat{n}$ & $|\hat{I}|$ & Cost/B  & RMSE & EQB        & Relative EQB \\
\hline
Experiment & 1,330     & 36          & 1       & 0.025285                                             & R\$562,323 & 1 \\
OGA        & 2,677     & 1           & 0.9939  & 0.018776                                             & R\$312,363 & 0.555\\
LASSO      & 2,762     & 0           & 0.99475 & 0.018789                                             & R\$313,853 & 0.558\\
POST-LASSO & 2,677     & 1           & 0.9939  & 0.018719                                             & R\$312,363 & 0.555\\
\hline\hline
\end{tabular}
\end{center}
\end{table}

\begin{table}[ht]
\caption{Day-care (outcome: health assessment)}\label{tab: daycare test4 summary}
 \begin{center}
\begin{tabular}{lcccccc}
\hline\hline
Method     & $\hat{n}$ & $|\hat{I}|$ & Cost/B  & RMSE & EQB        & Relative EQB \\
\hline
Experiment & 1,330     & 36          & 1       & 0.025442                                             & R\$562,323 & 1\\
OGA        & 2,762     & 0           & 0.99475 & 0.018799                                             & R\$308,201 & 0.548\\
LASSO      & 2,762     & 0           & 0.99475 & 0.018799                                             & R\$308,201 & 0.548\\
POST-LASSO & 2,677     & 1           & 0.9939  & 0.018735                                             & R\$306,557 & 0.545\\
\hline\hline
\end{tabular}
\end{center}
\end{table}

\paragraph{Results.}
Tables~\ref{tab: daycare test1 summary} and \ref{tab: daycare test4 summary}
summarize the results of the covariate selection procedures. For the cognitive test
outcome, OGA and POST-LASSO select one covariate (``$|\hat{I}|$''),\footnote{For OGA, it is an
indicator variable whether the respondent has finished secondary education, which is an
important predictor of outcomes; for POST-LASSO, it is the number of rooms in the house,
which can be considered as a proxy for wealth of the household, and again, an important
predictor of outcomes.} whereas LASSO does not select any covariate. The selected sample
sizes (``$\hat{n}$'') are 2,677 for OGA and POST-LASSO, and 2,762 for LASSO, which are almost twice
as large as the actual sample size in the experiment. The performance of the three covariate selection methods in terms of the precision of the resulting treatment effect estimator is measured by the square-root value of the minimized MSE criterion function (``RMSE'') from Step~2 of our procedure. We focus on the MSE, but notice that gains in MSE translate into gains in the power of the corresponding t-test as discussed in Section~\ref{sec:problem}. The three methods perform similarly well and improve precision by about 25\% relative to the experiment. Also, all three methods manage
to exhaust the budget, as indicated by the cost-to-budget ratios (``Cost/B'') close
to one. We do not put any strong emphasis on the selected covariates as the improvement of
the criterion function is minimal relative to the case that no covariate is selected (i.e., the
selection with LASSO). The results for the health assessment outcome are very similar to
those of the cognitive test with POST-LASSO selecting one variable (the number of rooms
in the house), whereas OGA and LASSO do not select any covariate.

To assess the economic gain of having performed the covariate selection procedure after the
first wave, we include the column ``EQB'' (abbreviation of ``equivalent budget'') in
Tables~\ref{tab: daycare test1 summary} and \ref{tab: daycare test4 summary}. The first
entry of this column in Table~\ref{tab: daycare test1 summary} reports the budget
necessary for the selection of $\hat{n}=$ 1,330 and all covariates, as was carried out in the
experiment. For the three covariate selection procedures, the column shows the budget that
would have sufficed to achieve the same precision as the actual experiment in terms of the minimum value of the MSE criterion function in Step~2. For
example, for the cognitive test outcome, using the OGA to select the sample size and the
covariates, a budget of R\$312,363  would have sufficed to achieve the experimental
RMSE of $0.025285$. This is a huge reduction of costs by about 45 percent, as
shown in the last column called ``relative EQB''. Similar reductions in costs are possible
when using the LASSO procedures and also when considering the health assessment
outcome.

In Appendix~F, we perform an out-of-sample evaluation by splitting the dataset into training samples for the covariate selection step and evaluation samples for the computation of the performance measures RMSE and EQB. The results are very similar to those in Tables~\ref{tab: daycare test1 summary} and \ref{tab: daycare test4 summary}.

Appendix~D presents the results of Monte Carlo simulations that mimic this dataset, and
shows that all three methods select more covariates and smaller sample sizes as we increase
the predictive power of some covariates. This finding suggests that the covariates collected in the survey were not predicting the outcome very well and,
therefore, in the next wave the researcher should spend more of the available budget to
collect data on more individuals, with no (or only a minimal) household survey.
Alternatively,  the researcher may want to redesign the household survey to include
questions whose answers are likely better predictors of the outcome.

\subsection{Provision of School Grants in Senegal}
\label{sec: school grants}

In this subsection, we consider the study by \citet{Carneiro-et-al:14} who evaluate, using an
RCT, the impact of school grants on student learning in Senegal.
The authors collect original data not only on the treatment status of schools (treatment and
control) and on student learning, but also on a variety of household, principal, and teacher
characteristics that could potentially affect learning. 


The dataset contains two waves, a baseline and a follow-up, which we use for the study of two different hypothetical scenarios. In the first scenario, the researcher has access to a
pre-experimental dataset consisting of all outcomes and covariates collected in the baseline
survey of this experiment, but not the follow-up data. The researcher applies the covariate selection procedure to this
pre-experimental dataset to find the optimal sample size and set of covariates for the randomized control trial to be carried out after the first wave.
In the second scenario, in addition to the pre-experimental sample from the first wave the researcher now also has access to the post-experimental outcomes collected in the follow-up (second wave). In this second scenario, we treat the follow-up outcomes as the outcomes of interest and include baseline outcomes in the pool of covariates that predict follow-up outcomes.

As in the previous subsection, we calibrate the cost function based on the full dataset from the experiment, but for solving the survey design problem we focus on a subset of individuals and variables from the original questionnaire. For simplicity, we exclude all household
variables from the analysis, because they were only collected for 4 out of the
12 students tested in each school, and we remove covariates whose sample variance is equal
to zero. Again, for simplicity, of the four outcomes (math test, French test, oral test,
and receptive vocabulary) in the original experiment, we only consider the first one
(math test) as our outcome variable. We drop individuals for whom at least one answer in the
survey or the outcome variable is missing. This sample selection procedure leads to sample
sizes of $N=2,280$ for the baseline math test outcome. For the second scenario discussed above where we use also the follow-up outcome, the
sample size is  smaller ($N=762$) because of non-response in the follow-up outcome and because we restrict the sample to the control group of the follow-up. In the first scenario in which we predict the baseline outcome, dropping household variables reduces the original number of covariates in the survey from 255 to $M=142$. The remaining covariates are school- and teacher-level variables. In the second scenario in which we predict follow-up outcomes, we add the three baseline outcomes to the covariate pool, but at the same time remove two covariates because they have variance zero when restricted to the control group. Therefore, there are $M=143$ covariates in the second scenario.

\paragraph{Calibration of the cost function.}
We specify the cost function \eqref{eq: cost decomposition} with components \eqref{eq:
interview blocked}--\eqref{eq: train blocked} to model the data collection procedure as
implemented in \citet{Carneiro-et-al:14}. Each school forms a cluster. We calibrate the
parameters using the costs faced by the researchers and their actual budgets for training,
administrative, and interview costs. The total budget for one wave of data collection in this
experiment, excluding the costs of the household survey, was approximately \$192,200.

For the calibration of the cost function, we use the original sample size, the original number
of covariates in the survey (except those in the household survey), and the original number
of outcomes collected at baseline. The three baseline outcomes were much more expensive
to collect than the remaining covariates. In the second scenario, we therefore group the former together as high-cost
variables, and all remaining covariates as low-cost variables. Appendix~B provides a detailed description of all components of the calibrated cost function.

\paragraph{Implementation.}
The implementation of the covariate selection procedures is identical to the one in the
previous subsection except that   we consider here two different specifications of the
pre-experimental sample $\mathcal{S}_{\rm pre}$, depending on whether the outcome of
interest is the baseline or follow-up outcome.

\paragraph{Results.}

\begin{table}[!ht]
\caption{School grants (outcome: math test)} \label{tab: schoolgrants testm summary}
 \begin{center}
\begin{tabular}{lcccccc}
\hline\hline
Method     & $\hat{n}$ & $|\hat{I}|$ & Cost/B  & RMSE & EQB        & Relative EQB\\
\hline\\
\multicolumn{7}{c}{(a) Baseline outcome}\\[3pt]
experiment & 2,280 & 142 & 1       & 0.0042272 & \$30,767 & 1\\
OGA        & 3,018 & 14  & 0.99966 & 0.003916  & \$28,141 & 0.91\\
LASSO      & 2,985 & 18  & 0.99968 & 0.0039727 & \$28,669 & 0.93\\
POST-LASSO & 2,985 & 18  & 0.99968 & 0.0038931 & \$27,990 & 0.91\\ [6pt]

\multicolumn{7}{c}{(b) Follow-up outcome}\\ [3pt]
experiment & 762  & 143 & 1       & 0.0051298 & \$52,604 & 1\\
OGA        & 6,755 & 0   & 0.99961 & 0.0027047 & \$22,761 & 0.43\\
LASSO      & 6,755 & 0   & 0.99961 & 0.0027047 & \$22,761 & 0.43\\
POST-LASSO & 6,755 & 0   & 0.99961 & 0.0027047 & \$22,761 & 0.43\\ [6pt]

\multicolumn{7}{c}{(c) Follow-up outcome, no high-cost covariates}\\[3pt]
experiment & 762  & 143 & 1       & 0.0051298 & \$52,604 & 1\\
OGA        & 5,411 & 140 & 0.99879 & 0.0024969 & \$21,740 & 0.41\\
LASSO      & 5,444 & 136 & 0.99908 & 0.00249   & \$22,082 & 0.42\\
POST-LASSO & 6,197 & 43  & 0.99933 & 0.0024624 & \$21,636 & 0.41\\ [6pt]

\multicolumn{7}{c}{(d) Follow-up outcome, force baseline outcome}\\[3pt]
experiment & 762  & 143 & 1       & 0.0051298 & \$52,604 & 1\\
OGA        & 1,314 & 133 & 0.99963 & 0.0040293 & \$41,256 & 0.78\\
LASSO      & 2,789 & 1   & 0.9929  & 0.0043604 & \$42,815 & 0.81\\
POST-LASSO & 2,789 & 1   & 0.9929  & 0.0032823 & \$32,190 & 0.61\\
\hline\hline
\end{tabular}
\end{center}
\end{table}

Table~\ref{tab: schoolgrants testm summary} summarizes the results of the covariate
selection procedures. Panel (a) shows the results of the first scenario in which the baseline
math test is used as the outcome variable to be predicted. Panel (b) shows the
corresponding results for the second scenario in which the baseline outcomes are treated as
high-cost covariates and the follow-up math test is used as the outcome to be predicted.

For the baseline outcome in panel (a), the OGA selects only $|\hat{I}|=14$ out of the
$145$ covariates with a selected sample size of $\hat{n}=3,018$, which is about 32\% larger
than the actual sample size in the experiment. The results for the LASSO and POST-LASSO
methods are similar. As in the previous subsection, we measure the performance of the three covariate selection methods by the estimated precision of the resulting treatment effect estimator (``RMSE''). As in the previous section focus on the MSE, but notice that gains in MSE translate into gains in the power of the corresponding t-test as discussed in Section~\ref{sec:problem}. The three methods improve the precision by about 7\% relative to the experiment. Also, all three methods manage to essentially
exhaust the budget, as indicated by cost-to-budget ratios (``Cost/B'') close to one. As in the previous subsection, we measure the economic gains from using the covariate selection procedures by the equivalent budget (``EQB'') that each of the method requires to achieve the precision of the experiment. All three methods require equivalent budgets that are 7-9\% lower than that of the experiment.

All variables that the OGA selects as strong predictors of baseline outcome are plausibly related
to student performance on a math test:\footnote{Appendix E shows the
full list and definitions of selected covariates for the baseline outcome.} They are related to
important aspects of the community surrounding the school (e.g., distance to the nearest
city), school equipment (e.g., number of computers), school infrastructure (e.g., number of
temporary structures), human resources (e.g., teacher--student ratio, teacher training), and
teacher and principal perceptions about which factors are central for success in the school
and about which factors are the most important obstacles to school success. 

For the follow-up outcome in panel (b), the budget used in the experiment increases due to the addition of the three expensive baseline outcomes to the pool of covariates.
All three methods select no covariates and exhaust the budget by using the maximum feasible sample size of 6,755, which is almost nine times larger than the sample size in the experiment. The implied precision of the treatment effect estimator improves by about 47\% relative to the experiment, which translates into the covariate selection methods requiring less than half of the experimental budget to achieve the same precision as in the experiment. These are substantial statistical and economic gains from using our proposed procedure.

\paragraph{Sensitivity Checks.} In RCT's, baseline outcomes tend to be strong predictors of the follow-up outcome. One may therefore be concerned that, because the OGA first selects the most predictive covariates which in this application are also much more expensive than the remaining low-cost covariates, the algorithm never examines what would happen to the estimator's MSE if it first selects the most predictive low-cost covariates instead. In principle, such selection could lead to a lower MSE than any selection that includes the very expensive baseline outcomes. As a sensitivity check we therefore perform the covariate selection procedures on the pool of covariates that excludes the three baseline outcomes. Panel (c) shows the corresponding results. In this case, all methods indeed select more covariates and smaller sample sizes than in panel (b), and achieve a slightly smaller MSE. The budget reductions relative to the experiment as measured by EQB are also almost identical to those in panel (b). Therefore, both selections of either no covariates and large sample size (panel (b)) and many low-cost covariates with somewhat smaller sample size (panel (c)) yield very similar and significant improvements in precision or significant reductions in the experimental budget, respectively.\footnote{Note that there is no sense in which need to be concerned about identification of the minimizing set of covariates. There may indeed exist several combinations of covariates that yield similar precision of the resulting treatment effect estimator. Our objective is highest possible precision without any direct interest in the identities of the covariates that achieve that minimum.}

As discussed in Section~\ref{subsec: discussion}, one may want to ensure balance of the control and treatment group, especially in terms of strong predictors such as baseline outcomes. Checking balance requires collection of the relevant covariates. Therefore, we also perform the three covariate selection procedures when we force each of them to include the baseline math outcome as a covariate. In the OGA, we can force the selection of a covariate by performing group OGA as described in Section~\ref{sec:algorithm}, where each group contains a low-cost covariate together with the baseline math outcome. For the LASSO procedures, we simply perform the LASSO algorithms after partialing out the baseline math outcome from the follow-up outcome.
The corresponding results are reported in panel (d). Since baseline outcomes are very expensive covariates, the selected sample sizes relative to those in panels (b) and (c) are much smaller. OGA selects a sample size of 1,314 which is almost twice as large as the experimental sample size, but about 4-5 times smaller than the OGA selections in panels (b) and (c). In contrast to OGA, the two LASSO procedures do not select any other covariates beyond the baseline math outcome. As a result of forcing the selection of the baseline outcome, all three methods achieve an improvement in precision, or reduction of budgets respectively, of around 20\% relative to the experiment. These are still substantial gains, but the requirement of checking balance on the expensive baseline outcome comes at the cost of smaller improvements in precision due to our procedure.

In Appendix~F, we also perform an out-of-sample evaluation for this application by splitting the dataset into training samples for the covariate selection step and evaluation samples for the computation of the performance measures RMSE and EQB. The results are qualitatively similar to those in Table~\ref{tab: schoolgrants testm summary}.

\section{Relation to the Existing Literature}\label{sec:literature}

In this section, we discuss related papers in the literature. We emphasize that the research
question in our paper is different from those studied in the literature and that our paper is a
complement to the existing work.

In the context of experimental economics, \citet{List-et-al:11} suggest several simple
rules of thumb that researchers can apply to improve the efficiency of their experimental
designs. They discuss the issue of experimental costs and estimation efficiency but did not
consider the problem of selecting covariates.

\citet{hahn2011we} consider the design of a two-stage experiment for estimating an
average treatment effect, and proposed to select the propensity score that minimizes the
asymptotic variance bound for estimating the average treatment effect. Their
recommendation is to assign individuals randomly between the treatment and control groups
in the second stage, according to the optimized propensity score. They use the covariate
information collected in the first stage  to compute the optimized propensity score.

\citet{Bhattacharya:Dupas:2012} consider the problem of allocating a binary treatment
under a budget constraint. Their budget constraint limits what fraction of the population can
be treated, and hence is different from our budget constraint. They discuss the costs of
using a large number of covariates in the context of treatment assignment.

\citet{McKenzie2012} demonstrates that taking multiple measurements of the outcomes
after an experiment can improve power under the budget constraint. His choice problem is
how to allocate a fixed budget over multiple surveys between a baseline and follow-ups.
The main source of the improvement in his case comes from taking repeated measures of
outcomes; see \citet{Frison:Pocock:1992} for this point in the context of clinical trials. In
the set-up of \citet{McKenzie2012}, a baseline survey measuring the outcome is especially
useful when there is high autocorrelation in outcomes. This would be analogous in our
paper to devoting part of the budget to the collection of a baseline covariate, which is highly
correlated with the outcome (in this case, the baseline value of the outcome), instead of just
selecting a post-treatment sample size that is as large as the budget allows for. In this way,
\citet{McKenzie2012} is perhaps closest to our paper in spirit.

In a very recent paper, \citet{Dominitz:Manski:16} proposed the use of statistical decision
theory to study allocation of a predetermined budget between two sampling processes of
outcomes: a high-cost process of good data quality and a low-cost process with
non-response or low-resolution interval measurement of outcomes.
Their main concern is data quality between two sampling processes and is distinct from our main focus, namely the simultaneous selection of the set of covariates and the sample size.

\section{Concluding Remarks}\label{sec: conclusion}

We develop data-driven methods for designing a survey in a randomized experiment
using information from a pre-existing dataset. Our procedure is optimal in a sense that it minimizes the mean squared error
of the average treatment effect estimator and maximizes the power of the corresponding t-test, and can handle a large number of potential covariates as
well as complex budget constraints faced by the researcher. We have illustrated the
usefulness of our approach by showing substantial improvements in precision of the resulting estimator or substantial reductions in the researcher's budget in two empirical applications.

We recognize that there are several other potential reasons guiding the choice of covariates in a survey. These may be as important as the one we focus on, which is the precision of the treatment effect estimator. We show that it is possible and important to develop practical tools to help researchers make such decisions. We regard our paper as part of the broader task of making the research design process more rigorous and transparent.

Some important issues remain as interesting future research topics. 
For example, we have assumed that the pre-experimental sample $\mathcal{S}_{\rm pre}$ is large, and therefore the
difference between the minimization of the sample average and that of the population
expectation is negligible. However, if the sample size of $\mathcal{S}_{\rm pre}$ is small
(e.g., in a pilot study), one may be concerned about over-fitting, in the sense of selecting too
many covariates. A straightforward solution would be to add a term to the objective
function that penalizes a large number of covariates via some information criteria (e.g., the Akaike information criterion (AIC) or the Bayesian information criterion (BIC)).

{\singlespacing

\ifx\undefined\BySame
\newcommand{\BySame}{\leavevmode\rule[.5ex]{3em}{.5pt}\ }
\fi \ifx\undefined\textsc
\newcommand{\textsc}[1]{{\sc #1}}
\newcommand{\emph}[1]{{\em #1\/}}
\let\tmpsmall\small
\renewcommand{\small}{\tmpsmall\sc}
\fi

}



\clearpage

\renewcommand\thepage{A-\arabic{page}} \setcounter{page}{1}


\appendix

\section*{Appendix A: Large-Budget Properties of the Algorithm}\label{sec:properties}
\renewcommand{\thesection}{A}

\renewcommand{\thetheorem}{A.\arabic{theorem}}
\setcounter{theorem}{0}
\renewcommand{\theremark}{A.\arabic{remark}}
\setcounter{remark}{0}
\renewcommand{\theequation}{A.\arabic{equation}}
      \setcounter{equation}{0}

In this appendix, we provide non-asymptotic bounds on the empirical risk of the OGA
approximation $\hat{\mathbf{f}} := \mathbf{Z}\hat{\gamma}$. Following
\citet{barron2008kk}, we define
 $$\|f\|_{\mL_1^N} := \inf \Big\{ \sum_{k=1}^p |\beta_k|_2:\, \beta_k \in\mathbb{R}^{|G_{k}|}\text{ and }
 f=\sum_{k=1}^p X_{G_{k}}'\beta_k \Big\}. $$
When the expression $f=\sum_{k=1}^p X_{G_{k}}'\beta_k$ is not unique, we take the true
$f$ to be one with the minimum value of $\|f\|_{\mL_1^N}$. This gives $f:= \gamma_0'X$
and $\mathbf{f}:= \mathbf{X} \gamma_0$ for some $\gamma_0$. Note that $\mathbf{f}$ is
defined by $\mathbf{X}$ with the true parameter value $\gamma_0$, while
$\hat{\mathbf{f}}$ is an OGA estimator of $\mathbf{f}$ using only $\mathbf{Z}$. The
following theorem bounds the finite sample approximation to the MSE of the treatment
effect estimator
$$\widehat{MSE}_{\hat{n},N}(\hat{\mathbf{f}}) := \| \bY - \hat{\mathbf{f}}\|_N^2/\hat{n},$$
which is equal to the objective function in \eqref{eq: sample problem}. Note that
$\widehat{MSE}_{\hat{n},N}(\hat{\mathbf{f}})$ can also be called the ``empirical risk''.

The following theorem is a modification of Theorem~2.3 of \citet{barron2008kk}. Our
result is different from \citet{barron2008kk} in two respects: (i) we pay explicit attention to
the group structure, and (ii) our budget constraint is different from their termination rule.

\begin{theorem}\label{thm: risk bound}
Assume that $(\bX_{G_j}' \bX_{G_j})/N = \mathbf{I}_{|G_j|}$ for each $j=1,\dots,p$.
Suppose $\mathcal{N}$ is a finite subset of $\mathbb{N}_+$, $c:\{0,1\}^M\times
\mathbb{N}_+\to \mathbb{R}$ some function, and $B>0$ some constant. Then the
following bound holds:
	\begin{equation}
		\widehat{MSE}_{\hat{n},N}(\hat{\mathbf{f}}) - \widehat{MSE}_{\hat{n},N}(\mathbf{f}) \leq \frac{4 \|f\|^2_{\mL_1^N}}{\hat{n}}  \left(\frac{1}{\min\{p, k_{\hat{n}} \}} \right).
	\end{equation}
\end{theorem}

The theorem provides a non-asymptotic bound on the empirical risk of the OGA
approximation, but the bound also immediately yields asymptotic consistency in the
following sense. Suppose $\|f\|_{\mL_1^N}<\infty$ (Remark~\ref{rem: finite L1n-norm}
discusses this condition). Then, the empirical risk of $\hat{\mathbf{f}}$ is asymptotically
equivalent to that of the true predictor $\mathbf{f}$ either if the selected sample size
$\hat{n}\to\infty$ or if both the total number of groups $p$ and the number of selected
groups $k_{\hat{n}}$ diverge to infinity. Consider, for example, the simple case in which,
for a given sample size $n$, data collection on every covariate incurs the same costs (i.e.,
$\tilde{c}(n)$) and each group consists of a single covariate. Then the total data collection
costs are equal to the number of covariates selected multiplied by $\tilde{c}(n)$ (i.e., $c(S,n)
= \tilde{c}(n) \sum_{j=1}^M S_j$). Assuming that $\tilde{c}(n)$ is non-decreasing in $n$,
we then have
$$
\frac{1}{\hat{n} \min\{p, k_{\hat{n}} \}} =
\frac{1}{\hat{n} \min\{M,\hat{m}\}} = \frac{1}{\hat{n} \min\left\{M, \lfloor B/\tilde{c}(\hat{n}) \rfloor \right\}}, $$
where $\lfloor x \rfloor$ denotes the largest integer smaller than $x$. Therefore, we obtain
consistency if $\hat{n}M\to_p \infty$ and $\hat{n}B / \tilde{c}(\hat{n}) \to_p \infty$.
Continue to assume that $\tilde{c}(n)$ is increasing in $n$ and $\mathcal{N}$ contains
sample sizes bounded away from zero. Then, both rate conditions are satisfied, for example,
as the budget increases, $B\to\infty$, and the costs per covariate does not increase faster than
linearly in the sample size.\footnote{In fact, the costs could be allowed to increase with $n$
at any rate as long as $B\to\infty$ at a faster rate, so that we have $\hat{n}B /
\tilde{c}(\hat{n}) \to_p \infty$.} Note that consistency can hold irrespectively of whether
the number of covariates $M$ is finite or infinite.

\begin{remark}\label{rem: finite L1n-norm}
The condition $\|f\|_{\mL_1^N}<\infty$ is trivially satisfied when $p$ is finite. In the case $p
\to \infty$,  the condition $\|f\|_{\mL_1^N}<\infty$ requires that not all groups of covariates
are equally important in the sense that the coefficients $\beta_k$, when their $\ell_2$ norms
are sorted in decreasing order, need to converge to zero fast enough to guarantee that
$\sum_{k=1}^\infty |\beta_k|_2 < \infty$.
\end{remark}

\begin{remark}
If suitable laws of large numbers apply, we can also replace the condition
$\|f\|_{\mL_1^N}<\infty$ by its population counterpart.
\end{remark}

\begin{remark}
The minimal sample size $n_0$ in $\mathcal{N}$ could, for example, be determined by
power calculations \citep[see, e.g.][]{duflo2007, McConnell2015} that guarantee a certain
power level for an hypothesis test of $\beta=0$.
\end{remark}

\subsection*{Proofs}

\begin{proof}[Proof of Lemma \ref{lem: equiv MSE}]
Let $U_i(\gamma) := Y_i - \gamma' X_i - \beta_0 D_i$.
The homoskedastic error assumption implies that conditional on $D_1,\ldots,D_n$, the estimator $\hat{\beta}(\gamma,n)$ is unbiased and thus the
finite-sample MSE of $\hat{\beta}(\gamma)$ is
\begin{eqnarray*}
&&\text{Var}\left(\hat{\beta} (\gamma, n) \given D_1,\ldots,D_n\right) \\
&&\quad= \frac{1}{n} \text{Var} \left(Y_i - {\gamma}'X_i\given D_i=0\right)
\Big\{  \Big( n^{-1} \sum_{i=1}^n D_i \Big) \Big(1 - n^{-1} \sum_{i=1}^n D_i \Big) \Big\}^{-1},
\end{eqnarray*}
which equals the desired expression because $\text{Var} (Y_i - {\gamma}'X_i | D_i=0) = \text{Var} (Y_i - {\gamma}'X_i | D_i=1) = \sigma^2(\gamma)$.
\end{proof}

\begin{proof}[Proof of Lemma~\ref{lem: equiv power}]
	First, notice that as in Lemma~\ref{lem: equiv MSE},
	$$\hat{\beta}(\gamma,n) - \beta_0 \;|\; (D_1,\ldots,D_n) \;\sim\; N\left(0, \frac{\sigma^2(\gamma)}{n \bar{D}_n(1-\bar{D}_n)}   \right). $$
	Therefore:
	\begin{align*}
		&P_{\beta}\left(\hat{t}(\gamma,n) > z \given D_1,\ldots,D_n\right)\\
		&\qquad= P_{\beta}\left(\frac{ \hat{\beta}(\gamma,n)}{\sigma(\gamma)/ \sqrt{n\bar{D}_n(1-\bar{D}_n)}} > z \given D_1,\ldots,D_n\right)\\
			&\qquad= P_{\beta}\left(\frac{ \hat{\beta}(\gamma,n)-\beta}{\sigma(\gamma) / \sqrt{n\bar{D}_n(1-\bar{D}_n)}} > z -\frac{ \beta}{\sigma(\gamma) / \sqrt{n\bar{D}_n(1-\bar{D}_n)}} \given D_1,\ldots,D_n\right)\\
			&\qquad= \Phi\left( \frac{ \beta}{\sigma(\gamma) / \sqrt{n\bar{D}_n(1-\bar{D}_n)}} -z\right)
	\end{align*}
	The power of the two-sided test then is
	\begin{align*}
		&P_{\beta}\left(\left|\hat{t}(\gamma,n)\right| > c_{1-\alpha/2} \given D_1,\ldots,D_n\right)\\
			&\qquad= \Phi\left( \frac{ \beta}{\sigma(\gamma) / \sqrt{n\bar{D}_n(1-\bar{D}_n)}} -c_{1-\alpha/2}\right) + 1 - \Phi\left( \frac{ \beta}{\sigma(\gamma) / \sqrt{n\bar{D}_n(1-\bar{D}_n)}} +c_{1-\alpha/2}\right)
	\end{align*}
	Differentiating this expression with respect to $\sigma(\gamma)$ yields
	\begin{align*}
		&\frac{\partial}{\partial \sigma(\gamma)} P_{\beta}\left(\left|\hat{t}(\gamma,n)\right| > c_{1-\alpha/2} \given D_1,\ldots,D_n\right)\\
			&\qquad= \phi\left( \frac{ \beta}{\sigma(\gamma) / \sqrt{n\bar{D}_n(1-\bar{D}_n)}} -c_{1-\alpha/2}\right) \frac{ -\beta}{\sigma^2(\gamma) / \sqrt{n\bar{D}_n(1-\bar{D}_n)}}\\
			&\qquad\quad - \phi\left( \frac{ \beta}{\sigma(\gamma) / \sqrt{n\bar{D}_n(1-\bar{D}_n)}} +c_{1-\alpha/2}\right) \frac{ -\beta}{\sigma^2(\gamma) / \sqrt{n\bar{D}_n(1-\bar{D}_n)}}\\
			&\qquad= \left[ \phi\left( \frac{ \beta}{\sigma(\gamma) / \sqrt{n\bar{D}_n(1-\bar{D}_n)}} +c_{1-\alpha/2}\right) - \phi\left( \frac{ \beta}{\sigma(\gamma) / \sqrt{n\bar{D}_n(1-\bar{D}_n)}} -c_{1-\alpha/2}\right)\right]\\
			&\qquad\quad \times \frac{ \beta}{\sigma^2(\gamma) / \sqrt{n\bar{D}_n(1-\bar{D}_n)}}\\
			&\qquad\leq 0
	\end{align*}
	where the inequality follows because the expression in the square brackets has the same sign as $-\beta$.
\end{proof}

\begin{proof}[Proof of Theorem~\ref{thm: risk bound}]
This theorem can be proved by arguments similar to those used in the proof of Theorem~2.3
in \citet{barron2008kk}. In the subsequent arguments, we fix $n$ and leave indexing by $n$
implicit.

First, letting $\hat{\br}_{k-1,i}$ denote the $i$th component of $\hat{\br}_{k-1}$, we have
	\begin{eqnarray*}
		\| \hat{\br}_{k-1} \|_N^2 &=& N^{-1} \sum_{i=1}^N \hat{\br}_{k-1,i} Y_i\\
				&=& N^{-1} \sum_{i=1}^N \hat{\br}_{k-1,i} U_i + N^{-1} \sum_{i=1}^N \hat{\br}_{k-1,i} \sum_{j=1}^\infty X_{G_{j},i}'\beta_j\\
				&\leq& \| \hat{\br}_{k-1} \|_{N} \Big\| \mathbf{Y} -  \sum_{k=1}^\infty \bX_{G_{k}}'\beta_k \Big\|_{N}
				+ \Big[ \sum_{j=1}^\infty | \beta_j |_2 \Big]  N^{-1} | \hat{\br}_{k-1}' \bX_{G_k} |_2 \\
				&\leq& \frac{1}{2} \Big( \| \hat{\br}_{k-1} \|_{N}^2 + \Big\| \mathbf{Y} -  \sum_{k=1}^\infty \bX_{G_{k}}'\beta_k \Big\|_{N}^2 \Big)
				+ \Big[ \sum_{j=1}^\infty | \beta_j |_2 \Big]  N^{-1} | \hat{\br}_{k-1}' \bX_{G_k} |_2,
	\end{eqnarray*}
	which implies that
	\begin{equation}\label{eq2-new}
	\| \hat{\br}_{k-1} \|_N^2
	- \Big\| \mathbf{Y} -  \sum_{k=1}^\infty \bX_{G_{k}}'\beta_k \Big\|_{N}^2
	\leq 	2\Big[ \sum_{j=1}^\infty | \beta_j |_2 \Big]  N^{-1} | \hat{\br}_{k-1}' \bX_{G_k} |_2.
	\end{equation}
	Note that if the left-hand side of \eqref{eq2-new} is negative for some $k=k_0$, then
	the conclusion of the theorem follows immediately for all $m \geq k_0 - 1$.
	Hence, we assume that  the left-hand side of \eqref{eq2-new} is positive, implying that
	\begin{equation}\label{eq2-new-1}
	\Big( \| \hat{\br}_{k-1} \|_N^2
	- \Big\| \mathbf{Y} -  \sum_{k=1}^\infty \bX_{G_{k}}'\beta_k \Big\|_{N}^2 \Big)^2
	\leq	4\Big[ \sum_{j=1}^\infty | \beta_j |_2 \Big]^2  N^{-2} | \hat{\br}_{k-1}' \bX_{G_k} |_2^2.
	\end{equation}

Let $P_k$ denote the projection matrix $P_k := \bX_{G_k} (\bX_{G_k}' \bX_{G_k})^{-1}
\bX_{G_k}' =N^{-1} \bX_{G_k}  \bX_{G_k}' $, where the second equality comes from the
assumption that $(\bX_{G_k}' \bX_{G_k})/N = \mathbf{I}_{|G_k|}$. Hence, it follows
from the fact that $P_k$ is the projection matrix that
	\begin{equation}\label{step01}
	\| \hat{\br}_{k-1} - P_k \hat{\br}_{k-1} \|_N^2 = \| \hat{\br}_{k-1}  \|_N^2 - \|  P_k \hat{\br}_{k-1} \|_N^2.
	\end{equation}
	Because $\hat{\br}_{k}$ is the best approximation to $\bY$ from $\hat{\mathcal{I}}_{n,k}$, we have
	\begin{equation}\label{step02}
	\| \hat{\br}_{k}  \|_N^2 \leq \| \hat{\br}_{k-1} - P_k \hat{\br}_{k-1} \|_N^2.
	\end{equation}
	Combining \eqref{step02} with \eqref{step01} and using the fact that $P_k^2 = P_k$, we have
	\begin{eqnarray}\label{eq1}
	\| \hat{\br}_{k} \|^2_N &\leq& \| \hat{\br}_{k-1} \|^2_N  -  \|  P_k \hat{\br}_{k-1} \|^2_N \nonumber\\
	&=& \| \hat{\br}_{k-1} \|^2_N  -  \|  N^{-1} \bX_{G_k}  \bX_{G_k}' \hat{\br}_{k-1} \|^2_N \nonumber\\
	&=& \| \hat{\br}_{k-1} \|^2_N  -  N^{-2} | \hat{\br}_{k-1}' \bX_{G_k} |_2^2,
	\end{eqnarray}

Now, combining \eqref{eq1} and \eqref{eq2-new-1} together yields
	\begin{equation}\label{induction-eq1}
	\| \hat{\br}_{k} \|^2_N \leq \| \hat{\br}_{k-1} \|^2_N-
	\frac{1}{4} \Big( \| \hat{\br}_{k-1} \|_N^2	-
	 \Big\| \mathbf{Y} -  \sum_{k=1}^\infty \bX_{G_{k}}'\beta_k \Big\|_{N}^2 \Big)^2
	 \Big[ \sum_{j=1}^\infty | \beta_j |_2 \Big]^{-2}.
	\end{equation}
	As in the proof of Theorem~2.3 in \citet{barron2008kk},
	let $a_k :=  \| \hat{\br}_{k} \|_N^2	-
	 \| \mathbf{Y} -  \sum_{k=1}^\infty \bX_{G_{k}}'\beta_k \|_{N}^2$.
	Then  \eqref{induction-eq1} can be rewritten as
	\begin{equation}\label{induction-eq2}
	a_k \leq a_{k-1} \Big( 1 - 	\frac{a_{k-1} }{4}  \Big[ \sum_{j=1}^\infty | \beta_j |_2 \Big]^{-2} \Big).
	\end{equation}
	Then the induction method used in  the proof of Theorem~2.1 in \cite{barron2008kk}
	gives the desired result, provided that
	$a_1 \leq 4 [ \sum_{j=1}^\infty | \beta_j |_2 ]^{2}$.
	As discussed at the end of the proof of Theorem~2.3 in \citet{barron2008kk}, the initial
condition is satisfied if  $a_0 \leq 4 [ \sum_{j=1}^\infty | \beta_j |_2 ]^{2}$. If not, we have
that $a_0 > 4 [ \sum_{j=1}^\infty | \beta_j |_2 ]^{2}$, which implies that $a_1 < 0$ by
\eqref{induction-eq2}. Hence, in this case, we have that
	$	\| \hat{\br}_{1} \|_N^2 	\leq
	 \| \mathbf{Y} -  \sum_{k=1}^\infty \bX_{G_{k}}'\beta_k \|_{N}^2$
	 for which there is nothing else to prove.

Then, we have proved that the error of the group OGA satisfies
\begin{equation*}
\| \hat{\br}_{m} \|^2_N \leq
\Big\| \mathbf{Y} -  \sum_{k=1}^p \bX_{G_{k}}'\beta_k \Big\|^2_N +
\frac{4}{m}  \Big[ \sum_{j=1}^p | \beta_j |_2 \Big]^{2}, \; m=1,2,\ldots.
\end{equation*}
Equivalently, we have, for any $n\in\mathcal{N}$ and any $k \geq 1$,
	$$\| \bY - \hat{\mathbf{f}}_{n,k}\|_N^2 - \| \bY - {\mathbf{f}}\|_N^2 \leq \frac{4
\|f\|^2_{\mL_1^N}}{k}.$$ %
Because $\mathcal{N}$ is a finite set, the desired result immediately follows by substituting
in the definition of $\hat{\mathbf{f}}$ and $k_{\hat{n}}$.
\end{proof}

\section*{Appendix B: Cost Functions Used in Section \ref{sec:app}}\label{sec:cost:detail}
\renewcommand{\thesection}{B}
\renewcommand{\theequation}{B.\arabic{equation}}
      \setcounter{equation}{0}

In this appendix, we provide detailed descriptions of the cost functions used in Section
\ref{sec:app}.

\subsection*{Calibration of the Cost Function in Section \ref{sec: daycare}}\label{sec:cost:detail:daycare}

Here, we give a detailed description of components of the cost function used in Section
\ref{sec: daycare}.

\begin{itemize}
	\item \textbf{Administration costs.}
		The administration costs in the survey were R\$10,000 and the average survey took
two hours per household to conduct (i.e., $T(S)=120$ measured in minutes). Therefore,
		\[
			c_{\rm admin}(S,n) = \phi ( 120) ^{\alpha}=\hbox{10,000}.
		\]
		If we assume that, say, $\alpha=0.4$ (which means that the costs of 60
		minutes are about 75.8 percent of the costs of 120 minutes), we obtain $\phi \approx$ 1,473.

	\item \textbf{Training costs.} The training costs in the survey were R\$25,000, that is,
		$$c_{\rm train}(S,n) = \kappa(1,466) \cdot 120 = \hbox{25,000},$$
		so that $\kappa(1,466)  \approx 208$. It is reasonable to assume that there exists
some lumpiness in the training costs. For example, there could be some indivisibility in
hotel rooms that are rented, and in the number of trainers required for each training
session. To reflect this lumpiness, we assume that
		\[
		\kappa( n) =\left\{
		\begin{array}{ccc}
		150 & \text{if} & 0<n\leq \hbox{1,400} \\
		208 & \text{if} & \hbox{1,400}<n\leq \hbox{3,000} \\
		250 & \text{if} & \hbox{3,000}<n\leq \hbox{4,500} \\
		300 & \text{if} & \hbox{4,500}<n\leq \hbox{6,000} \\
		350 & \text{if} & \hbox{6,000}<n%
		\end{array}%
		\right..
		\]%
Note that, in this specification, $\kappa( \hbox{1,466})  \approx 208$, as calculated
above. We take this as a point of departure to calibrate $\kappa( n)$. Increases in sample
size $n$ are likely to translated into increases in the required number of field workers for
the survey, which in turn lead to higher training costs. Our experience in the field (based
on running surveys in different settings, and on looking at different budgets for different
versions of this same survey) suggests that, in our example, there is some concavity in this
cost function, because an increase in the sample size, in principle, will not require a
proportional increase in the number of interviewers, and an increase in the number of
interviewers will probably require a less than proportion increase in training costs. For
example, we assume that a large increase in the size of the sample, from 1,500 to 6,000,
leads to an increase in  $\kappa( n)$ from 208 to 300 (i.e., an increase in overall training
costs of about 50 percent).

	\item \textbf{Interview costs.}
		Interview costs were R\$630,000, accounting for the majority of the total survey costs, that is,
		$$c_{\rm interv}(S,n) = \hbox{1,466} \cdot \eta + \hbox{1,466} \cdot p\cdot 120 = \hbox{630,000}, $$
so that $\eta +120p \approx 429.74.$  The costs of traveling to each household in this
survey were approximately half of the total costs of each interview. If we choose $\eta
=200$, then the fixed costs $\eta $ amount to about 47 percent of the total interview costs,
which is consistent with the actual costs of the survey. Then we obtain the price per unit
of survey time as $p\approx 1.91$. It is also reasonable to assume that half of the variable
costs per individual are due to the collection of the three outcomes in the survey, because
their administration was quite lengthy. The costs of collecting the outcomes could also be
seen as fixed costs (equal to $0.955\times 120=114.6$), which means that the price per
unit of survey time for each of the remaining covariates is about $0.955$. In sum, we can
rewrite interview costs as
		$$c_{\rm interv}(S,n) = \hbox{1,466} \times (200 + 114.6) + 1,466 \times 0.955\times 120 = \hbox{630,000}. $$

	\item \textbf{Price per covariate.} We treat the sample obtained from the original
experiment as $\mathcal{S}_{\rm pre}$, a pilot study or the first wave of a data
collection process, based on which we want to decide which covariates and what sample
size to collect in the next wave. We perform the selection procedure for each outcome
variable separately, and thus adjust $T(S) = \tau (1+\sum_{j=1}^M S_j)$. For simplicity,
we assume that to ask each question on the questionnaire takes the same time, so that
$\tau_0 = \tau_j = \tau$ for every question; therefore, $T(S)= \tau
(1+\sum_{j=1}^{M}S_{j}) = 120$. Note that we set $\tau_0 = \tau$ here, but the high
costs of collecting the outcome variables are reflected in the specification of $\eta$ above. 		
This results in $\tau = 120/(1+\sum_{j=1}^{M}S_{j})$. The actual number of covariates
collected in the experiment was $40$; so $\sum_{j=1}^{M}S_{j} = 40$, and thus $\tau
\approx 3$.

\item \textbf{Rescaled budget.} Because we use only a subsample of the original
    experimental sample, we also scale down the original budget of R\$665,000 down to
    R\$569,074, which corresponds to the costs of selecting all 36 covariates in the
    subsample; that is, $c(\mathbf{1}, \hbox{1,330})$ where $\mathbf{1}$ is a
    36-dimensional vector of ones and $c(S,n)$ is the calibrated cost function.

\end{itemize}

\subsection*{Calibration of the Cost Function in Section \ref{sec: school grants}}\label{sec:cost:detail:school grants}

Here, we present a detailed description of components of the cost function used in Section
\ref{sec: school grants}.

\begin{itemize}
	\item \textbf{Administration costs.} The administration costs for the low- and high-cost
covariates were estimated to be about \$5,000 and \$24,000, respectively. The high-cost
covariates were four tests that took about 15 minutes each (i.e., $T_{\rm high}(S) = 60$).
For the low-cost covariates (teacher and principal survey), the total survey time was
around 60 minutes, so $T_{\rm low}(S) = 60$. High- and low-cost variables were
collected by two different sets of enumerators, with different levels of training and skills.
Therefore,
		\[
			\phi_{\rm low} ( 60) ^{\alpha_{\rm low}} = \hbox{5,000}\qquad\text{and}\qquad
\phi_{\rm high} ( 60) ^{\alpha_{\rm high}} = \hbox{24,000}.
		\]
If we assume that, say, $\alpha_{\rm low}=\alpha_{\rm high}=0.7$, we obtain $\phi_{\rm
low} \approx 285$ and $\phi_{\rm high} \approx$ 1,366.

	\item \textbf{Training costs.} $\mu_{\rm high}$ and $\mu_{\rm low}$ are the numbers
of enumerators collecting high- and low-cost variables, respectively. The training costs for
enumerators in the high and low groups increase by $20$ for each set of additional $20$
low-cost enumerators, and by $12$ for each set of $4$ high-cost enumerators:
		$$\kappa_{\rm low}(c,n_c):= 20 \sum_{k=1}^{19} k\cdot \ind\{20(k-1) < \mu_{\rm low}(c,n_c) \leq 20k \} $$
		and
		$$\kappa_{\rm high}(c,n_c):= 12 \sum_{k=1}^{17} k\cdot \ind\{4(k-1) < \mu_{\rm high}(c,n_c) \leq 4k \}. $$
This is reasonable because enumerators for low-cost variables can be trained in large
groups (i.e., groups of 20), while enumerators for high-cost variables need to be trained in
small groups (i.e., groups of 4). However, training a larger group demands a larger room,
and, in our experience, more time in the room. The lumpiness comes from the costs of
hotel rooms and the time of the trainers. The numbers 20 and 12 as the average costs of
each cluster of enumerators were chosen based on our experience with this survey (even
if the design of the training and the organization of the survey was not exactly the same as
the stylized version presented here), and reflect both the time of the trainer and the costs
of hotel rooms for each type of enumerators. Because the low-cost variables are
questionnaires administered to principals and teachers, in principle the number of
required enumerators only depends on $c$ (i.e., $\mu_{\rm low}(c,n_c) = \lfloor
\lambda_{\rm low} c \rfloor$). High-cost variables are collected from students, and
therefore the number of required enumerators should depend on $c$ and $n_c$, so
$\mu_{\rm high}(c,n_c) = \lfloor \lambda_{\rm high} c \mu_{n,{\rm high}}(n_c)\rfloor$.
We assume that the latter increases again in steps, in this case of $10$ individuals per
cluster, that is,
		$$\mu_{n,{\rm high}}(n_c):= \sum_{k=1}^7 k\cdot \ind\{10(k-1) < n_c \leq 10 k \}. $$
We let $\lambda_{\rm low}=0.14$ (capturing the idea that one interviewer could do
about seven schools) and $\lambda_{\rm high}=0.019$ (capturing the idea that one
enumerator could perhaps work with about 50 children). The training costs in the survey
were \$1,600 for the low-cost group of covariates and \$1,600 for the high-cost group of
covariates.

	\item \textbf{Interview costs.} We estimate that interview costs in the survey were
\$150,000 and \$10,000 for the high- and low-cost variables, respectively, i.e.
		$$ \psi_{\rm low}(350) \eta_{\rm low} + 350 \cdot p_{\rm low}\cdot 60 = \hbox{10,000} $$
		and
		$$ \psi_{\rm high}(350, 24) \eta_{\rm high} + 350 \cdot 24 p_{\rm high}\cdot 60 = \hbox{150,000}. $$
We set $\psi_{\rm low}(c) = \mu_{\rm low}(c)$ and $\psi_{\rm high}(c,n_c) = \mu_{\rm
high}(c,n_c)$, the number of required enumerators for the two groups, so that $\eta_{\rm
low}$ and $\eta_{\rm high}$ can be interpreted as fixed costs per enumerator. From the
specification of $\mu_{\rm low}(c)$ and $\mu_{\rm high}(c,n_c)$ above, we obtain
$\mu_{\rm low}(350)=50$ and $\mu_{\rm high}(350,24)=20$. The fixed costs in the
survey were about $\psi_{\rm low}(350) \eta_{\rm low}=500$ and $\psi_{\rm
high}(350,24) \eta_{\rm high}=$ 1,000 for low- and high-cost covariates. Therefore,
$\eta_{\rm low} = 500/50=10$ and $\eta_{\rm high} = \hbox{1,000}/20=50$. Finally,
we can solve for the prices $p_{\rm low} = (\hbox{10,000}-500)/(350\times 60) \approx
0.45$ and $p_{\rm high} = (\hbox{150,000}-\hbox{1,000})/(350\times 24\times 60)
\approx 0.3$.

	\item \textbf{Price per covariate.} For simplicity, we assume that to ask each low-cost
question takes the same time, so that $\tau_j = \tau_{\rm low}$ for every low-cost
question (i.e., $j=1,\ldots,M_{\rm low}$), and that each high-cost question takes the
same time (i.e., $\tau_j=\tau_{\rm high}$) for all $j=M_{\rm low}+1,\ldots,M$. The
experimental budget contains funding for the collection of one outcome variable, the
high-cost test results at follow-up, and three high-cost covariates at baseline. We modify
$T_{\rm high}(S)$ accordingly:  $T_{\rm high}(S) = \tau_{\rm high}(1+\sum_{j=M_{\rm
low}+1}^M S_j)=4\tau_{\rm high}$ so that $\tau_{\rm high} = 60/4 = 15$. Similarly,
originally there were $255$ low-cost covariates, which leads to $\tau_{\rm low} =
120/255\approx 0.47$.

\item \textbf{Rescaled budget.} As in the previous subsection, we use only a subsample
    of the original experimental sample. Therefore,  we scale down the original budget  to
    the amount that corresponds to the costs of collecting all covariates used in the
    subsample. As a consequence, the rescaled budget is \$25,338 in the case of baseline
    outcomes and \$33,281 in the case of the follow-up outcomes. 	
	
\end{itemize}

\section*{Appendix C: A Simple Formulation of the Problem}\label{sec:simpleproblem}
\renewcommand{\thesection}{C}
\renewcommand{\theequation}{C.\arabic{equation}}
      \setcounter{equation}{0}

\subsection*{Uniform Cost per Covariate}

Take the following simple example where: (1) all covariates are orthogonal to each other;
(2) all covariates have the same price, and the budget constraint is just $B=nk$, where $n$ is
sample size and $k$ is the number of covariates. Order the covariates by the contribution to
the MSE, so that the problem is to choose the first $k$ covariates (and the corresponding
$n$).

Define $\sigma ^{2}( k) = ({1}/{N}) \sum_{i=1}^{N}( Y_{i}-\gamma_{0,k} ^{\prime }X_{i})
^{2}$, where $\gamma_{0,k}$ is the same as the vector of true coefficients 
$\gamma_0$ except that all coefficients after the $(k+1)$th coefficient are set to be zeros,  and let $MSE( k,n) = ({1}/{n}) \sigma ^{2}( k) $. For the convenience of using simple
calculus, suppose that $k$ is continuous, ignoring that $k$ is a positive integer, and that
$\sigma ^{2}( k)$ is twice continuously differentiable. This would be a reasonable
first-order approximation when there are a large number of covariates, which is our set-up in
the paper. Because we ordered the covariates by the magnitude of their contribution to a
reduction in the MSE, we have $\partial \sigma ^{2}( k)/{\partial k}<0$, and ${\partial^{2}\sigma ^{2}(k)}/{\partial k^{2}}>0$.

The problem we solve in this case is just
\[
\min_{n,k}\frac{1}{n}\sigma ^{2}( k) \quad {\rm s.t.}\quad nk\leq B.
\]%
Assume we have an interior solution and that  $n$ is also continuous. Replace the budget
constraint in the objective function and we obtain
\[
\min_{n,k}\frac{k}{B}\sigma ^{2}( k).
\]%
This means that $k$ is determined by
\[
\sigma ^{2}( k) +k\frac{\partial \sigma ^{2}( k) }{\partial k}=0,
\]%
or%
\begin{align}\label{foc-appendix}
\frac{\sigma ^{2}( k) }{k}+\frac{\partial \sigma ^{2}(k) }{\partial k}=0,
\end{align}
which in this particular case does not depend on $B$.  Then, $n$ is given by the budget
constraint (i.e., $n = B/k$).

Another way to see where this condition comes from is just to start from the budget
constraint. If we want to always satisfy it then, starting from a particular choice of $n$ and
$k$ yields
\[
n \cdot dk+k \cdot dn=0,
\]
or
\[
\frac{dn}{dk}=-\frac{n}{k}.
\]

Now, suppose we want to see what happens when $k$ increases by a small amount. In that
case, keeping $n$ fixed, the objective function falls by
\[
\frac{1}{n}\frac{\partial \sigma ^{2}( k) }{\partial k}dk.
\]
This is the marginal benefit of increasing $k$. However, $n$ cannot stay fixed, and needs to
decrease by $(n/k) dk$ to keep the budget constraint satisfied. This means that the objective
function will increase by
\[
\left( - \frac{1}{n^{2}}\right) \sigma ^{2}\left( k\right) \left( -\frac{n}{k}\right) dk.
\]
This is the marginal cost of increasing $k$.

At the optimum, in an interior solution, marginal costs and marginal benefits need to
balance out, so
\[
\frac{1}{n k} \sigma ^{2}\left( k\right) dk=-\frac{1}{n}\frac{\partial \sigma ^{2}\left( k\right) }{\partial k}%
dk
\]%
or%
\[
\frac{\sigma ^{2}\left( k\right) }{k}+\frac{\partial \sigma ^{2}\left(
k\right) }{\partial k}=0,
\]%
which reproduces \eqref{foc-appendix}.

There are a few things to notice in this simple example.
\begin{enumerate}
\item[(1)] The marginal costs of an increase in $k$ are increasing in $\sigma
^{2}\left( k\right) $. This is because increases in $n$ are more important
role for the MSE when $\sigma ^{2}\left( k\right) $ is large than when it is
small.

\item[(2)] The marginal costs of an increase in $k$ are decreasing in $k$. This is
because when $k$ is large, adding an additional covariate does not cost much
in terms of reductions in $n$.

\item[(3)] A large $n$ affects the costs and benefits of increasing $k$ in similar way. Having a
    large $n$ reduces benefits of additional covariates because it dilutes the decrease in
    $\sigma ^{2}\left( k\right) $. Then, on one hand, it increases costs through the budget
    constraint, as a larger reduction in $n $ is needed to compensate for the same change in
    $k$. However, on the other hand, it reduces costs, because when $n$ is large, a
    particular reduction in $n$ makes much less difference for the MSE than in the case
    where $n$ is small.

\item[(4)] We can rewrite this condition as
\[
\frac{1}{k}+\frac{{\partial \sigma ^{2}(k)}/{\partial k}}{\sigma ^{2}\left( k\right) }=0,
\]%
where the term $({\partial \sigma ^{2}(k)}/{\partial k})/{\sigma ^{2}\left( k\right) }$ is
the percentage change in the unexplained variance from an increase in $k$.
\end{enumerate}

If we combine
\[
\frac{dn}{n}=\frac{dk}{k},
\]%
which comes from  the budget constraint,
and
\[
\frac{1}{MSE\left( n,k\right) }\frac{\partial MSE\left( n,k\right) }{%
\partial n}=-\frac{1}{n},
\]%
we notice that the percentage decrease in $MSE$ from an increase in $n$ is just
$({dn})/{n}$, the percentage change in $n$, which in turn is just equal to $({dk})/{k}$. So
what the condition above says is that we want to equate the percentage change in the
unexplained variance from a change in $k$ to the percentage change in the MSE from the
corresponding change in $n$.

Perhaps even more interesting is to notice that $k$ is the survey cost per individual in this
very simple example. Then this condition says that we want to choose $k$ to equate the
percentage change in the survey costs per individual ($({dk})/{k}$) to the percentage change
in the residual
variance
\[
\frac{{\partial \sigma ^{2}(k)}/{\partial k}}{\sigma ^{2}\left( k\right) }dk.\]
This condition explicitly links the impacts of $k$ on the survey costs and on the reduction
in the MSE.

Adding fixed costs $F$ of visiting each individual is both useful and easy in this very simple
framework. Suppose there are a fixed costs $F$ of going to each individual, so the budget
constraint is $n\left( F+k\right) =B$. Proceeding as above, we can rewrite our problem as
\[
\min_{n,k}\frac{F+k}{B}\sigma ^{2}\left( k\right).
\]%
This means that $k$ is determined by
\[
\sigma ^{2}\left( k\right) +\left( F+k\right) \frac{\partial \sigma^{2}\left( k\right) }{\partial k}=0,
\]%
or%
\begin{eqnarray*}
\frac{1}{F+k}+\frac{{\partial \sigma ^{2}\left( k\right) }/{\partial k}}{\sigma ^{2}\left( k\right) } &=&0.
\end{eqnarray*}
Note that, when there are large fixed costs of visiting each individual, increasing $k$ is not
going to be that costly at the margin. It makes it much easier to pick a positive $k$.
However, other than that, the main lessons (1)--(4) of this simple model remain unchanged.

\subsection*{Variable Cost per Covariate}

If covariates do not have uniform costs, then the problem is much more complicated.
Consider again a simple set-up where all the regressors are orthogonal, and we order them
by their contribution to the MSE. However, suppose that the magnitude of each covariate's
contribution the MSE takes a discrete finite number of values. Let $\mathcal{R}$ denote the set of these discrete values. Let $r$ denote an element of 
$\mathcal{R}$ and $R = |\mathcal{R}|$ (the total number of all elements in $\mathcal{R}$). 
 There are
many potential covariates within each $r$ group, each with a different price $p$. The
support of $p$ could be different for each $r$. So, within each $r$, we will then order
variables by $p$. The problem will be to determine the optimal $k$ for each $r$ group. Let 
$\mathbf{k} \equiv \{ k_{r}: r \in \mathcal{R} \}$.

The problem is
\[
\min_{n,\mathbf{k}}\frac{1}{n}\sigma ^{2}\left(
\mathbf{k}\right) \quad{\rm s.t.}
\sum_{r \in \mathcal{R}} c_r(k_r) \leq B,
\]%
where $c_{r}\left( k_{r}\right) =\sum_{l=1}^{k_{r}}p_{l}$ are the costs of variables of type
$r$ used in the survey. We can also write it as $ c_{r}\left( k_{r}\right) =p_{r}\left(
k_{r}\right) k_{r}$, where $ p_{r}\left( k_{r}\right)
=({\sum_{l=1}^{k_{r}}p_{l}})/{k_{r}}$. Because we order the variables by price (from
low to high), $\partial p_{r}\left( k_{r}\right)/{\partial k_{r}}>0$. Let $\sigma
_{r}^{2}={\partial \sigma ^{2}\left( \mathbf{k}\right)}/{\partial k_{r}}$, which
is a constant (this is what defines a group of variables).

Then, assume we can approximate $p_{l}\left( k_{r}\right) $ by a continuous function and
that we have an interior solution. Then, substituting the budget constraint in the objective
function:
\[
 \min_{n,\mathbf{k}} \frac{1}{B}\left[ \sum_{r \in \mathcal{R}} c_r(k_r) \right] \sigma ^{2}\left(
\mathbf{k}\right).
\]%
From the first-order condition for $k_{r}$,
\[
\frac{\partial c_{r}\left( k_{r}\right) }{\partial k_{r}}\sigma ^{2}\left(
\mathbf{k}\right) +\left[ \sum_{r \in \mathcal{R}} c_r(k_r) \right] \frac{\partial \sigma ^{2}\left(
\mathbf{k}\right) }{\partial k_{r}}=0,
\]%
or%
\[
\frac{{\partial c_{r}\left( k_{r}\right) }/{\partial k_{r}}}{
\sum_{r \in \mathcal{R}} c_r(k_r)   }=-\frac{
{\partial \sigma ^{2}\left( \mathbf{k}\right) }/{\partial k_{r}}}{\sigma ^{2}\left( \mathbf{k}\right) }.
\]
What this says is that, for each $r$, we choose variables up to the point
where the percent marginal contribution of the additional variable to the
residual variance equals the percent marginal contribution of the additional
variable to the costs per interview, just as in the previous subsection.

\section*{Appendix D: Simulations}\label{sec:sim}
\renewcommand{\thesection}{D}
\renewcommand{\theequation}{D.\arabic{equation}}
      \setcounter{equation}{0}
\renewcommand{\thetable}{D.\arabic{table}}
      \setcounter{table}{0}

In this appendix, we study the finite sample behavior of our proposed data collection
procedure, and compare its performance to other variable selection methods. We consider
the linear model from above, $Y=\gamma'X + \varepsilon$, and mimic the data-generating
process in the day-care application of Section~\ref{sec: daycare} with the cognitive test
outcome variable.

First, we use the dataset to regress $Y$ on $X$. Call the regression coefficients
$\hat{\gamma}_{\rm emp}$ and the residual variance $\hat{\sigma}_{\rm emp}^2$. Then,
we regress $Y$ on the treatment indicator to estimate the treatment effect $\hat{\beta}_{\rm
emp}=0.18656$. We use these three estimates to generate Monte Carlo samples as follows.
For the pre-experimental data $\mathcal{S}_{\rm pre}$, we resample $X$ from the
empirical distribution of the $M=36$ covariates in the dataset and generate outcome
variables by $Y=\gamma'X + \varepsilon$, where $\varepsilon \sim N(0,\hat{\sigma}_{\rm
emp}^2)$ and
$$\gamma = \hat{\gamma}_{\rm emp}+ \frac{1}{2} {\rm sign}(\hat{\gamma}_{\rm emp})\kappa\bar{\gamma}. $$
We vary the scaling parameter $\kappa \in \{0, 0.3, 0.7, 1\}$ and
$\bar{\gamma}:=(\bar{\gamma}_1,\ldots,\bar{\gamma}_{36})'$ is specified in three
different fashions, as follows:%
\begin{itemize}
	\item ``lin-sparse'', where the first five coefficients linearly decrease from 3 to 1, and all others are zero, that is,
		$$\bar{\gamma}_k := \left\{\begin{array}{cc} 3-2(k-1)/5,& 1\leq k \leq 5\\ 0,&\text{otherwise}
     \end{array}\right. ;$$
	\item ``lin-exp'', where the first five coefficients linearly decrease from 3 to 1, and the remaining decay exponentially, that is,
		$$\bar{\gamma}_k := \left\{\begin{array}{cc} 3-2(k-1)/5,& 1\leq k \leq 5\\ e^{-k},&k>5  \end{array}\right. ;$$
	\item ``exp'', where exponential decay $\bar{\gamma}_k := 10 e^{-k}$.
\end{itemize}

\begin{figure}[!b]
	\centering
	\includegraphics[width=0.75\textwidth]{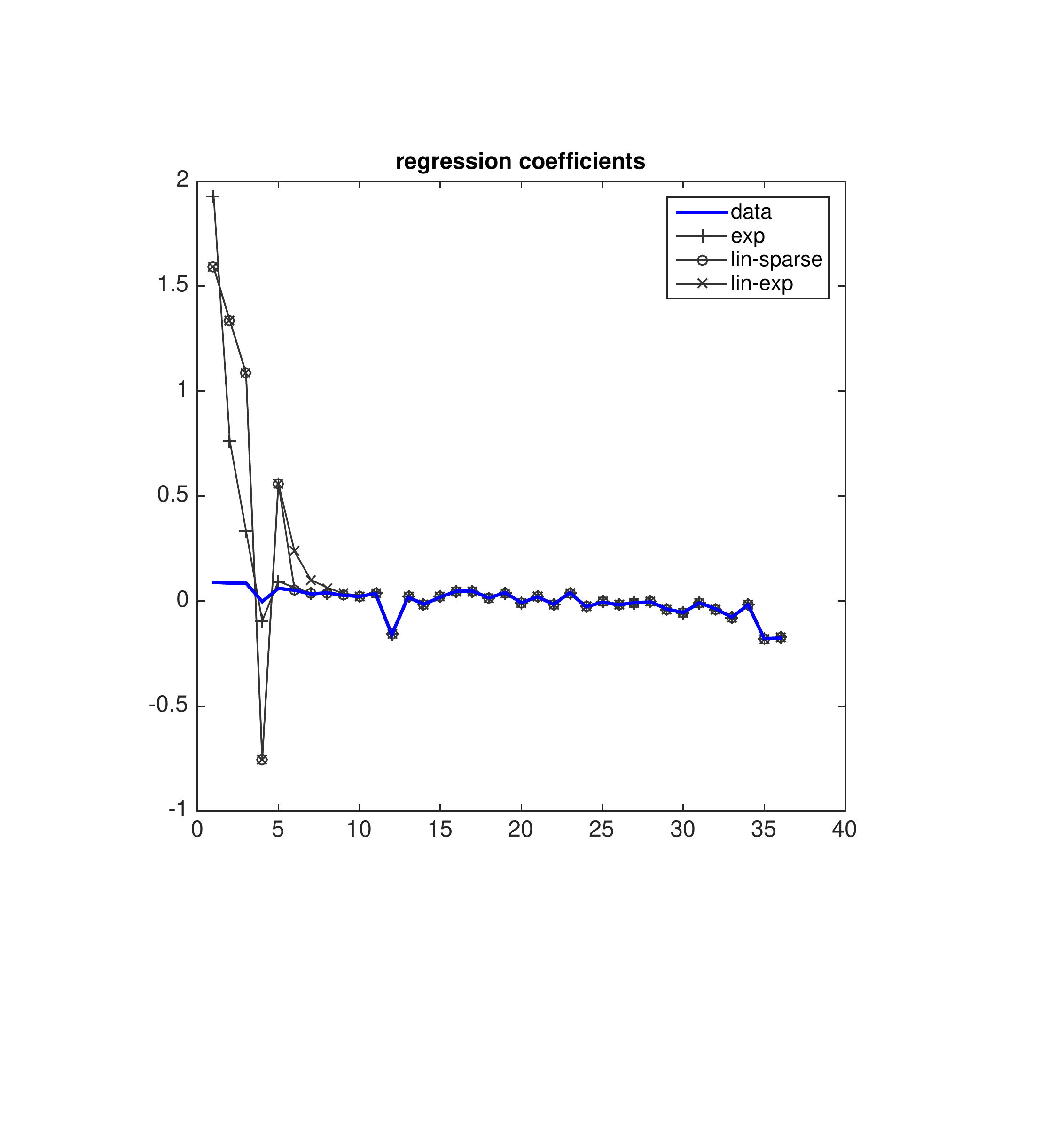}	
	\caption{Regression coefficients in the simulation when $\kappa=0.3$}\label{fig: reg coefficients}
\end{figure}

When $\kappa=0$, the regression coefficients $\gamma$ are equal to those in the empirical
application. When $\kappa>0$, we add one of the three specifications of $\bar{\gamma}$ to
the coefficients found in the dataset, thereby increasing (in absolute value) the first few
coefficients\footnote{Because all estimated coefficients in the dataset ($\hat{\gamma}_{\rm
emp}$) are close to zero and roughly of the same magnitude, we simply pick the first five
covariates that have the highest correlation with the outcome variable.} more than the
others, and thus increasing the importance of the corresponding regressors for prediction of
the outcome. Figure~\ref{fig: reg coefficients} displays the regression coefficients in the
dataset (i.e., when $\kappa=0$, denoted by the blue line labeled ``data''), and $\gamma$ for
the three different specifications when $\kappa=0.3$.

\begin{sidewaystable}[htbp]
\caption{Simulation results: lin-sparse}\label{tab: lin-sparse}
 \begin{center}
\begin{tabular}{llccccccccc}
\hline\hline
Scale & Method     & $\hat{n}$ & $|\hat{I}|$ & Cost/B  & $\sqrt{\widehat{MSE}_{\hat{n},N}(\hat{\mathbf{f}})}$
& $bias(\hat{\beta})$ & $sd(\hat{\beta})$ & RMSE($\hat{\beta}$) & EQB\\
\hline
0     & Experiment & 1,330      & 36          & 1       & 0.02498                                              & $-$0.0034284          & 0.049981          & 0.050048            & \$56,9074\\
      & OGA        & 2,508  & 1.4       & 0.99543 & 0.019249                                             & 0.00034598          & 0.038838          & 0.038801            & \$34,8586\\
      & LASSO      & 2,587  & 0.1       & 0.99278 & 0.019418                                             & 0.0020874           & 0.039372          & 0.039388            & \$35,6781\\
      & POST-LASSO & 2,529  & 1.0        & 0.99457 & 0.019222                                             & 0.00069394          & 0.037758          & 0.037727            & \$35,1659\\
0.3   & Experiment & 1,330      & 36          & 1       & 0.02494                                              & 0.00036992          & 0.049501          & 0.049453            & \$56,9074\\
      & OGA        & 2,350   & 3.9        & 0.99443 & 0.019751                                             & 0.0013905           & 0.038275          & 0.038262            & \$37,7076\\
      & LASSO      & 2,228  & 5.9       & 0.988   & 0.020346                                             & $-$0.00093089         & 0.041713          & 0.041682            & \$39,3017\\
      & POST-LASSO & 2,320  & 4.4       & 0.99321 & 0.019696                                             & $-$0.0020751          & 0.038746          & 0.038763            & \$37,1730\\
0.7   & Experiment & 1,330      & 36          & 1       & 0.024953                                             & $-$0.00086563         & 0.050992          & 0.050948            & \$56,9074\\
      & OGA        & 2,346      & 4.0           & 0.99433 & 0.020552                                             & $-$0.0020151          & 0.041475          & 0.041483            & \$39,7722\\
      & LASSO      & 2,218  & 6.1       & 0.98747 & 0.021145                                             & 0.00057516          & 0.043957          & 0.043917            & \$42,3971\\
      & POST-LASSO & 2,246  & 5.7       & 0.98929 & 0.020177                                             & 0.0019693           & 0.042683          & 0.042686            & \$38,7095\\
1     & Experiment & 1,330      & 36          & 1       & 0.024938                                             & $-$0.0021535          & 0.051146          & 0.05114             & \$56,9074\\
      & OGA        & 2,346      & 4.0           & 0.99433 & 0.021566                                             & 0.00044162          & 0.043389          & 0.043348            & \$43,8536\\
      & LASSO      & 2,172  & 6.9       & 0.98513 & 0.021383                                             & $-$0.00058106         & 0.045378          & 0.045336            & \$43,1589\\
      & POST-LASSO & 2,172  & 6.9       & 0.98513 & 0.019956                                             & $-$0.0048726          & 0.040967          & 0.041215            & \$38,0053\\\hline\hline
\end{tabular}
\end{center}
\end{sidewaystable}

\begin{sidewaystable}[htbp]
\caption{Simulation results: lin-exp}\label{tab: lin-exp}
 \begin{center}
\begin{tabular}{llccccccccc}
\hline\hline
Scale & Method     & $\hat{n}$ & $|\hat{I}|$ & Cost/B  & $\sqrt{\widehat{MSE}_{\hat{n},N}(\hat{\mathbf{f}})}$
& $bias(\hat{\beta})$ & $sd(\hat{\beta})$ & RMSE($\hat{\beta}$) & EQB\\
\hline
0     & Experiment & 1,330      & 36          & 1       & 0.024965                                             & 0.0027033           & 0.051564          & 0.051583            & \$569,074\\
      & OGA        & 2,509      & 1.3       & 0.99541 & 0.019249                                             & $-$0.00042961         & 0.03723           & 0.037195            & \$348,682\\
      & LASSO      & 2,588      & 0.1       & 0.99275 & 0.01941                                              & $-$0.003374           & 0.03845           & 0.03856             & \$357,261\\
      & POST-LASSO & 2,530      & 1.0           & 0.9946  & 0.019215                                             & 0.00076956          & 0.037924          & 0.037894            & \$351,755\\
0.3   & Experiment & 1,330      & 36          & 1       & 0.02492                                              & $-$0.0015645          & 0.049457          & 0.049432            & \$569,074\\
      & OGA        & 2,343      & 4.0       & 0.99421 & 0.019868                                             & $-$0.0014349          & 0.040197          & 0.040182            & \$379,540\\
      & LASSO      & 2,186      & 6.7        & 0.98569 & 0.020652                                             & 0.0019377           & 0.04084           & 0.040845            & \$403,004\\
      & POST-LASSO & 2,313      & 4.5       & 0.99288 & 0.019816                                             & $-$0.0025812          & 0.039587          & 0.039631            & \$377,876\\
0.7   & Experiment & 1,330      & 36          & 1       & 0.024936                                             & 0.0041527           & 0.050436          & 0.050556            & \$569,074\\
      & OGA        & 2,301      & 4.7       & 0.99247 & 0.020805                                             & $-$0.0017267          & 0.041303          & 0.041297            & \$408,990\\
      & LASSO      & 2,134      & 7.7       & 0.98551 & 0.02162                                              & $-$0.00071182         & 0.042716          & 0.042679            & \$440,232\\
      & POST-LASSO & 2,206      & 6.5       & 0.98955 & 0.020522                                             & 0.0013055           & 0.043358          & 0.043334            & \$400,219\\
1     & Experiment & 1,330      & 36          & 1       & 0.024964                                             & $-$0.0034064          & 0.049484          & 0.049551            & \$569,074\\
      & OGA        & 2,286      & 5.0        & 0.99187 & 0.021874                                             & $-$0.0025106          & 0.042304          & 0.042336            & \$451,756\\
      & LASSO      & 2,080      & 9.0       & 0.98793 & 0.021987                                             & $-$0.0015746          & 0.044218          & 0.044201            & \$454,765\\
      & POST-LASSO & 2,078      & 9.0       & 0.98787 & 0.020374                                             & 0.00077488          & 0.041977          & 0.041942            & \$396,218\\\hline\hline
\end{tabular}
\end{center}
\end{sidewaystable}

\begin{sidewaystable}[htbp]
\caption{Simulation results: exp}\label{tab: exp}
 \begin{center}
\begin{tabular}{llccccccccc}
\hline\hline
Scale & Method     & $\hat{n}$ & $|\hat{I}|$ & Cost/B  & $\sqrt{\widehat{MSE}_{\hat{n},N}(\hat{\mathbf{f}})}$
 & $bias(\hat{\beta})$ & $sd(\hat{\beta})$ & RMSE($\hat{\beta}$) & EQB\\
 \hline
0     & Experiment & 1,330      & 36          & 1       & 0.024953                                             & 0.00083077          & 0.054043          & 0.053996            & \$569,074\\
      & OGA        & 2,511      & 1.3         & 0.99538 & 0.019234                                             & 0.0016616           & 0.037237          & 0.037236            & \$348,426\\
      & LASSO      & 2,588      & 0.1         & 0.99278 & 0.019394                                             & $-$0.00049328         & 0.038849          & 0.038813            & \$356,941\\
      & POST-LASSO & 2,529      & 1.0         & 0.99452 & 0.019203                                             & $-$0.00044404         & 0.039549          & 0.039512            & \$351,403\\
0.3   & Experiment & 1,330      & 36          & 1       & 0.024947                                             & $-$0.00089522         & 0.051246          & 0.051202            & \$569,074\\
      & OGA        & 2,411      & 2.9         & 0.99605 & 0.019426                                             & 0.0016951           & 0.038729          & 0.038727            & \$359,950\\
      & LASSO      & 2,291      & 4.9         & 0.9911  & 0.020184                                             & $-$0.0022094          & 0.040243          & 0.040263            & \$389,560\\
      & POST-LASSO & 2,380      & 3.5         & 0.99514 & 0.019377                                             & 0.0014552           & 0.039996          & 0.039982            & \$359,662\\
0.7   & Experiment & 1,330      & 36          & 1       & 0.024946                                             & $-$0.0012694          & 0.050947          & 0.050912            & \$569,074\\
      & OGA        & 2,408      & 3.0         & 0.99605 & 0.019457                                             & 0.0015399           & 0.040789          & 0.040778            & \$362,287\\
      & LASSO      & 2,279      & 5.1         & 0.99039 & 0.020233                                             & 0.0011166           & 0.042491          & 0.042463            & \$391,128\\
      & POST-LASSO & 2,376      & 3.5         & 0.99515 & 0.019405                                             & $-$0.0023208          & 0.037252          & 0.037287            & \$361,903\\
1     & Experiment & 1,330      & 36          & 1       & 0.024948                                             & $-$0.0034014          & 0.051898          & 0.051957            & \$569,074\\
      & OGA        & 2,407      & 3.0         & 0.99603 & 0.019494                                             & 0.0022031           & 0.038846          & 0.038869            & \$364,015\\
      & LASSO      & 2,271      & 5.2         & 0.99008 & 0.020298                                             & 0.0016393           & 0.039024          & 0.039019            & \$392,857\\
      & POST-LASSO & 2,377      & 3.5         & 0.99516 & 0.019448                                             & $-$0.00085645         & 0.039135          & 0.039106            & \$363,023\\\hline\hline
\end{tabular}
\end{center}
\end{sidewaystable}

For each Monte Carlo sample from $\mathcal{S}_{\rm pre}$, we apply the OGA, LASSO,
and POST-LASSO methods, as explained in Section \ref{sec: daycare}.  The cost function
and budget are specified exactly as in the empirical application. We store the sample size
and covariate selection produced by each of the three procedures, and then mimic the
randomized experiment $\mathcal{S}_{\exp}$ by first 
drawing a new sample of $X$
from the same data-generating process as in $\mathcal{S}_{\rm pre}$. 
Then
we  generate random treatment indicators $D$, so that outcomes are determined by
$$Y = \hat{\beta}_{\rm emp} D + \gamma' X + \varepsilon,$$
where $\varepsilon$ is randomly drawn from $N(0,\hat{\sigma}_{\rm
emp}^2)$. 
 We then compute the treatment effect estimator $\hat{\beta}$ of $\beta$ as described in
Step 4 of Section~\ref{sec:problem}.

The results are based on $500$ Monte Carlo samples, $N=$ 1,330, which is the sample size
in the dataset, and $\mathcal{N}$ a fine grid from $500$ to 4,000. All covariates, those in
the dataset as well as the simulated ones, are studentized so that their variance is equal to
one.

For the different specifications of $\bar{\gamma}$, Tables~\ref{tab: lin-sparse}--\ref{tab:
exp} report the selected sample size ($\hat{n}$), the selected number of covariates
($|\hat{I}|$), the ratio of costs for that selection divided by the budget $B$, the square root of
the estimated MSE, $\sqrt{\widehat{MSE}_{\hat{n},N}(\hat{\mathbf{f}})}$, the bias and
standard deviation of the estimated average treatment effect ($bias(\hat{\beta})$ and
$sd(\hat{\beta})$), and the RMSE of $\hat{\beta}$ across the Monte Carlo samples of the
experiment.

Overall, all three methods perform similarly well across different designs and the number of
selected covariates tends to increase as $\kappa$ becomes large. No single method
dominates other methods, although POST-LASSO seems to perform slightly better than
LASSO. In view of the Monte Carlo results, we argue that the empirical findings reported in
Section \ref{sec: daycare} are likely to result from the lack of highly predictive covariates
in the empirical example.

\clearpage
\section*{Appendix E: Variables Selected in the School Grants Example}\label{sec:variabledesc}
\renewcommand{\thesection}{E}

\renewcommand{\thetable}{E.\arabic{table}}
      \setcounter{table}{0}

\begin{table}[!ht]
\caption{School grants (outcome: math test): selected covariates in panel (a) of
Table \ref{tab: schoolgrants testm summary}} \label{tab: schoolgrants testm varlist}
 \begin{center}
\begin{tabular}{llll}
\hline\hline
 OGA                          & LASSO                       & POST-LASSO\\
 \hline \\
 \text{Child is male}         & \text{Child is male}          & \text{Child is male}\\
 \text{Village pop.}          & \text{Dist. to Dakar}         & \text{Dist. to Dakar}\\
 \text{Piped water}           & \text{Dist. to city}          & \text{Dist. to city}\\
 \text{Teach-stud}            & \text{Village pop.}           & \text{Village pop.}\\
 \text{No. computers}         & \text{Piped water}            & \text{Piped water}\\
 \text{Req. (h) teach. qual.} & \text{No. computers}          & \text{No. computers}\\
 \text{Req. (h) teach. att.}  & \text{Req. (h) teach-stud}    & \text{Req. (h) teach-stud}\\
 \text{Obs. (h) manuals}      & \text{Hrs. tutoring}          & \text{Hrs. tutoring}\\
 \text{Books acq. last yr.}   & \text{Books acq. last yr.}    & \text{Books acq. last yr.}\\
 \text{Any parent transfer}   & \text{Provis. struct.}        & \text{Provis. struct.}\\
 \text{Teacher bacc. plus}    & \text{NGO cash cont.}         & \text{NGO cash cont.}\\
 \text{Teach. train. math}    & \text{Any parent transfer}    & \text{Any parent transfer}\\
 \text{Obst. (t) class size}  & \text{NGO promised cash}      & \text{NGO promised cash}\\
 \text{Measure. equip.}       & \text{Avg. teach. exp.}       & \text{Avg. teach. exp.}\\
 \text{ }                     & \text{Teacher bacc. plus}     & \text{Teacher bacc. plus}\\
 \text{ }                     & \text{Obs. (t) student will.} & \text{Obs. (t) student will.}\\
 \text{ }                     & \text{Obst. (t) class size}   & \text{Obst. (t) class size}\\
 \text{ }                     & \text{Silence kids}           & \text{Silence kids}
 \\ \hline\hline
\end{tabular}
\end{center}
\end{table}

\begin{table}[htbp]
\caption{Definition of variables in Table \ref{tab: schoolgrants testm varlist}}
 \begin{center}
\begin{tabular}{ll}
\hline
Variable               & Definition \\
\hline
Child is male          & Male student \\
Village pop.           & Size of the population in the village \\
Piped water            & School has access to piped water \\
Teach--stud            & Teacher--student ratio in the school \\
No. computers          & Number of computers in the school \\
Req. (h)\ teach. qual. & Principal believes teacher quality is a major requirement \\
                       & for school success \\
Req. (h) teach. att.   & Principal believes teacher attendance is a major requirement \\
                       & for school success \\
Obs. (h) manuals       & Principal believes the lack of manuals is a major obstacle\\
						& to school success \\
Books acq. last yr.    & Number of manuals acquired last year \\
Any parent transfer    & Cash contributions from parents \\
Teacher bacc. plus     & Teacher has at least a baccalaureate degree \\
Teach. train. math     & Teacher received special training in math \\
Obst. (t) class size   & Teacher believes class size is a major obstacle to school success \\
Measure. equip.        & There is measurement equipment in the classroom \\
Dist. to Dakar         & Distance to Dakar \\
Dist. to city          & Distance to the nearest city \\
Req. (h) teach--stud   & Principal believe teacher--student ratio is a major requirement \\
                       & for school success \\
Hrs. tutoring          & Hours of tutoring by teachers \\
Provis. struct.        & Number of provisional structures in school \\
NGO cash cont.         & Cash contributions by NGO \\
NGO promised cash      & Promised cash contributions by NGO \\
Avg. teach. exp.       & Average experience of teachers in the school \\
Obst. (t) student will.       & Teacher believes the lack of student willpower is one of the  \\
                       & main obstacles to learning in the school \\
Obst. (t) class size & Teacher believes the lack of classroom size is one of the main\\
						& obstacles to learning in the school\\
Silence kids           & Teacher has to silence kids frequently \\
\hline
\end{tabular}
\end{center}
\end{table}

\clearpage
\section*{Appendix F: Out-of-Sample Evaluations}\label{sec:oos eval}

In the empirical applications, we performed the covariate selection procedure as well as its evaluation (by RMSE and EQB) on the same pre-experimental sample. In this section, we study the sensitivity of our findings when the covariate selection and evaluation steps are performed on two separate samples.

We partition each of the datasets into five subsamples of equal size. Four of the five subsamples are merged to form the training set on which we perform the covariate selection procedure, and the remaining fifth subsample serves as evaluation sample on which we calculate the performance measures RMSE and EQB. Given the partition into five subsamples, there are five possible ways to combine them into training and evaluation samples. We perform the covariate selection on each of these five training samples using the same calibrated cost functions as in the main text, but adjusting the budget for the drop in sample size by letting the budget be the cost function $c(S,n)$ evaluated at the experimental selection $S=(1,\ldots,1)'$ and $n$ the length of the training sample. The output of the procedure consists of five sample size selections $\hat{n}$, five covariate selections, i.e. five values of $|\hat{I}|$, and five cost-to-budget ratios. Tables~\ref{tab: oos childcare cognitive}--\ref{tab: oos schoolgrants} show the averages of $\hat{n}$, $|\hat{I}|$, and ``Cost/B'' over those five different training samples. The RMSE is calculated using the estimate of $\gamma$ from the training sample and data on $Y$ and $X$ from the evaluation sample. Similarly, the EQB is the budget necessary to achieve the RMSE on the evaluation sample equal to that of the experiment when the covariate selection procedures are applied to the training sample. Both RMSE and EQB are then averaged over the five subsamples.

Overall, the results of this out-of-sample evaluation exercise are similar to those reported in the full-sample analysis of the main text. Qualitatively, in both applications, the covariate selection procedures recommend larger sample sizes than the experiment. The recommended sample size may differ somewhat from those reported in the main text because the budget and training sample size is different, but the orders of magnitude are the same. In the school grants application, we notice that the recommended number of covariates selected tends to be smaller than in the full-sample evaluation of the main text, but if anything the covariate selection procedures manage to achieve even lower relative equivalent budgets and lower RMSE than the experiment.

\begin{table}[ht]
 \begin{center}
 \caption{Day-care (outcome: cognitive test), 5-fold out-of-sample evaluation}
 \label{tab: oos childcare cognitive}
\begin{tabular}{rrrrrrr}
\hline\hline
method     & $\hat{n}$ & $|\hat{I}|$ & cost/B  & RMSE     & EQB				& relative EQB\\\hline
experiment & 1,330      & 36          & 1       & 0.029068 & R\$460,809.54 & 1\\
OGA        & 2,209      & 0.8         & 0.99503 & 0.020694 & R\$235,654.21 & 0.511\\
LASSO      & 2,260      & 0           & 0.99392 & 0.020777 & R\$237,425.38 & 0.515\\
POST-LASSO & 2,146      & 1.8         & 0.99464 & 0.020647 & R\$234,494.95 & 0.509\\\hline\hline
\end{tabular}
\end{center}

\end{table}

\begin{table}[ht]
 \begin{center}
 \caption{Day-care (outcome: health assessment), 5-fold out-of-sample evaluation}
 \label{tab: oos childcare health}
\begin{tabular}{rrrrrrr}
\hline\hline
method     & $\hat{n}$ & $|\hat{I}|$ & cost/B  & RMSE     & EQB & relative EQB\\\hline
experiment & 1,330      & 36          & 1       & 0.029313 & R\$460,809.54 & 1\\
OGA        & 2,221      & 0.6         & 0.99495 & 0.020708 & R\$232,066.95 & 0.504\\
LASSO      & 2,260      & 0           & 0.99392 & 0.020787 & R\$233,751.11 & 0.507\\
POST-LASSO & 2,158      & 1.6         & 0.9949  & 0.020644 & R\$231,224.87 & 0.502\\\hline\hline
\end{tabular}
\end{center}
\end{table}

\begin{table}[!ht]
\caption{School grants (outcome: math test), 5-fold out-of-sample evaluation}
\label{tab: oos schoolgrants}
 \begin{center}
\begin{tabular}{lcccccc}
\hline\hline
Method     & $\hat{n}$ & $|\hat{I}|$ & Cost/B  & RMSE & EQB        & Relative EQB\\
\hline\\
\multicolumn{7}{c}{(a) Baseline outcome}\\[3pt]
experiment & 1,824   & 142 & 1       & 0.0082721 & \$27,523.74 & 1\\
OGA        & 2,618.4 & 1.2 & 0.99823 & 0.0044229 & \$16,609.53 & 0.603\\
LASSO      & 2,658   & 0   & 0.9991  & 0.004445  & \$16,621.60 & 0.604\\
POST-LASSO & 2,638.2 & 1.2 & 0.99927 & 0.0044291 & \$16,651.77 & 0.605\\ [6pt]

\multicolumn{7}{c}{(b) Follow-up outcome}\\ [3pt]
experiment & 609  & 143 & 1       & 0.0098756 &  \$48,856.20 & 1\\
OGA        & 6,132 & 0   & 0.99885 & 0.0028432 & \$14,664.85 & 0.300\\
LASSO      & 6,132 & 0   & 0.99885 & 0.0028432 & \$14,664.85 & 0.300\\
POST-LASSO & 6,132 & 0   & 0.99885 & 0.0028432 & \$14,664.85 & 0.300\\ [6pt]

\multicolumn{7}{c}{(c) Follow-up outcome, no high-cost covariates}\\[3pt]
experiment & 609    & 143 & 1       & 0.0098756 &  \$48,856.20 & 1\\
OGA        & 6,092.8 & 0.8 & 0.99893 & 0.0027807 & \$14,571.10 & 0.298\\
LASSO      & 6,000.8 & 5.4 & 0.99905 & 0.002795  & \$14,651.75 & 0.300\\
POST-LASSO & 6,040.4 & 2.4 & 0.99887 & 0.0027623 & \$14,532.12 & 0.297\\ [6pt]

\multicolumn{7}{c}{(d) Follow-up outcome, force baseline outcome}\\[3pt]
experiment & 609    & 143 & 1       & 0.0098756 &  \$48,856.20 & 1\\
OGA        & 2,035.2 & 2.4 & 0.90783 & 0.0041647 & \$24,918.89 & 0.510\\
LASSO      & 2,494   & 1   & 0.99623 & 0.0046439 & \$25,522.72 & 0.522\\
POST-LASSO & 2,494   & 1   & 0.99623 & 0.0034893 & \$23,651.72 & 0.484\\
\hline\hline
\end{tabular}
\end{center}
\end{table}

\clearpage
\section*{Appendix G: The Case of Multivariate Outcomes}\label{sec:vector-outcome}
\renewcommand{\thesection}{G}
\renewcommand{\theequation}{G.\arabic{equation}}
      \setcounter{equation}{0}
\renewcommand{\thetable}{G.\arabic{table}}
      \setcounter{table}{0}

In this section, we consider an extension to the case of multivariate outcomes. 
If data on a particular regressor is collected, then the regressor is automatically available for regressions involving 
any of the outcomes. Therefore, it is natural to  select 
one common set of regressors for all outcomes.	Hence, 
our regression problem corresponds to  the special case of seemingly unrelated regressions (SUR) such that   
the vector of regressors is identical for each equation. In this case, it is well known that the OLS and GLS estimators  are algebraically identical. 
In other words, there is no loss of efficiency in using the single-equation OLS estimator even if regression errors are correlated.

Suppose there are $L$ outcome variables of interest, say $\{Y_{\ell, i}: \ell=1,\dots,L,   \text{ and }  i=1,\ldots, N \}$. Then 
a multivariate analog of \eqref{eq: sample problem} can be written as 
\begin{equation}\label{eq: sample problem multi}
	\min_{n\in \mathbb{N}_+,\, \bm{\gamma} = (\gamma_1',\ldots,
\gamma_L')' \in \mathbb{R}^{ML}} \frac{1}{nNL}\sum_{\ell =1 }^L \sum_{i=1}^N (Y_{\ell, i}-\gamma_\ell'X_i)^2\qquad \text{s.t.}\qquad  c(\mathcal{I}(\gamma),n) \leq B.
\end{equation}
In other words, the stacked version of the OLS problem is equivalent to regressing
$\bm{y} := (\bm{y}_1',\ldots,\bm{y}_L')'$ on $I_L \otimes \bm{X}$ conditional on the budget constraint, where 
$\bm{y}_\ell = (Y_{\ell,1},\ldots,Y_{\ell,L})'$, $I_L$ is the $L$-dimensional identity matrix, and  $\bm{X}$ is $N \times M$ dimensional matrix whose $i$th row is $X_i'$. 
Therefore, the OGA applies to this case as well with minor modifications. First, we need to redefine the outcome vector and the design matrix with the stacked  $\bm{y}$ and the enlarged design matrix $I_L \otimes \bm{X}$. Suppose that a variable selection problem is on individual components of $X_i$. Then note that because of the nature of the stacked regressions, we need to apply a group OGA with each group  consisting of $L$ columns of $[I_L \otimes \bm{X}]_k$, where $k = (\ell-1)M + m$ 
$(\ell = 1,\ldots,L)$ for each $m =1,\ldots,M$.

\end{document}